\documentclass[11pt]{article}

\usepackage[margin=0.88in]{geometry}
\usepackage{fontspec}
\defaultfontfeatures{Ligatures=TeX}
\newfontfamily\leanfont[
  Path=fonts/,
  UprightFont=STIXGeneral.otf,
  BoldFont=STIXGeneralBol.otf,
  ItalicFont=STIXGeneralItalic.otf,
  BoldItalicFont=STIXGeneralBolIta.otf,
  Scale=MatchLowercase
]{STIXGeneral}
\usepackage{microtype}
\usepackage{amsmath,amssymb,amsthm,mathtools}
\usepackage{booktabs,array,longtable,tabularx,multirow}
\usepackage[dvipsnames]{xcolor}
\usepackage{graphicx}
\usepackage{float}
\usepackage{tikz}
\usetikzlibrary{arrows.meta,positioning,fit,calc,decorations.pathreplacing,shapes.geometric}
\usepackage[most]{tcolorbox}
\usepackage{enumitem}
\usepackage{algorithm2e}
\usepackage{listings}
\usepackage{fancyhdr}
\usepackage[numbers,sort&compress]{natbib}
\usepackage{hyperref}
\usepackage[nameinlink,noabbrev]{cleveref}
\usepackage{url}
\usepackage{seqsplit}
\usepackage{caption}
\usepackage{subcaption}
\usepackage{setspace}

\definecolor{DeepBlue}{HTML}{174A7E}
\definecolor{SoftBlue}{HTML}{EAF2F8}
\definecolor{ProofGray}{HTML}{8A929B}
\definecolor{SoftGray}{HTML}{F4F5F6}
\definecolor{NucA}{HTML}{188038}
\definecolor{NucC}{HTML}{2962FF}
\definecolor{NucG}{HTML}{7B1FA2}
\definecolor{NucU}{HTML}{E37400}
\definecolor{AlertRed}{HTML}{B3261E}
\definecolor{GoodGreen}{HTML}{0B6B3A}

\hypersetup{
  colorlinks=true,
  linkcolor=DeepBlue,
  citecolor=DeepBlue,
  urlcolor=DeepBlue,
  pdftitle={Designability of RNA Targets with Up to Two Length-2 Helices},
  pdfauthor={Ashutosh S. Jogalekar}
}

\setlist{leftmargin=*,itemsep=2pt,topsep=4pt}
\newtheorem{theorem}{Theorem}[section]
\newtheorem{lemma}[theorem]{Lemma}
\newtheorem{proposition}[theorem]{Proposition}
\newtheorem{corollary}[theorem]{Corollary}
\newtheorem{definition}[theorem]{Definition}
\newtheorem{remark}[theorem]{Remark}

\newcommand{\A}{\textsf{A}}
\newcommand{\C}{\textsf{C}}
\newcommand{\G}{\textsf{G}}
\newcommand{\U}{\textsf{U}}
\newcommand{\Bcol}{\mathsf{B}}
\newcommand{\Wcol}{\mathsf{W}}
\newcommand{\Gray}{\mathsf{G}}

\newcommand{\Ktwo}{\mathcal{K}_{\le 2}}

\newcommand{\comp}{\operatorname{comp}}
\newcommand{\inv}{\operatorname{inv}}
\newcommand{\level}{\ell}
\newcommand{\entry}{e}
\newcommand{\shortcount}{s}
\newcommand{\Accept}{\operatorname{Accept}}

\newcommand{\hashtext}[1]{\texttt{\seqsplit{#1}}}
\newcommand{\RepositoryURL}{\url{https://github.com/ajogalekar/rna-at-most-two-short-helices}}
\newcommand{\ArchiveDOI}{\href{https://doi.org/10.5281/zenodo.22101755}{\texttt{10.5281/zenodo.22101755}}}
\newcommand{\ReleaseVersion}{1.0.3}
\newcommand{\TOneSequenceLiteral}{AGCUUGCAGGGCCC}
\newcommand{\TTwoSequenceLiteral}{GGGAGGAGCUUGCACCUGGGCCCCCC}

\newtcolorbox{contributionbox}{
  colback=SoftBlue,colframe=DeepBlue,boxrule=0.8pt,arc=2pt,
  left=7pt,right=7pt,top=6pt,bottom=6pt
}
\newtcolorbox{cautionbox}{
  colback=SoftGray,colframe=ProofGray,boxrule=0.6pt,arc=2pt,
  left=7pt,right=7pt,top=6pt,bottom=6pt
}
\newtcolorbox{verificationbox}{
  colback=green!4,colframe=GoodGreen,boxrule=0.7pt,arc=2pt,
  left=7pt,right=7pt,top=6pt,bottom=6pt
}

\title{\textbf{Designability of RNA Targets with Up to Two Length-2 Helices}}
\author{Ashutosh S. Jogalekar\\
\small Szilard Scientific, LLC, Palo Alto, CA\\
\small \texttt{ashjogalekar@gmail.com}}
\date{August 2026}

\begin{document}
\maketitle

\begin{abstract}
RNA inverse folding asks whether one can choose an RNA sequence whose unique optimal secondary structure is a prescribed target.  We study this problem in a simplified four-letter model: only Watson--Crick A--U and C--G pairs are allowed, structures contain no pseudoknots, adjacent nucleotides may pair, and lower energy means a larger number of base pairs.

Building on the separated-coloring certificate and parity device of Hale\v{s} et al. and the modulo-$m$ construction of Boury et al., earlier work guarantees designability when every helix contains at least three base pairs.  Shorter helices are more difficult because they provide less freedom for coordinating the base-pair choices made at neighboring loops.  We prove that designability is nevertheless guaranteed when the target contains at most two helices of length 2, contains no helix of length 1, avoids the two standard local obstruction motifs $m_5$ and $m_{3\bullet}$, and has all remaining helices of length at least 3.

The proof builds on Boury et al.'s local helix-coloring transfers and adds a global counting argument showing that the demands created by at most two short helices can always be coordinated.

From the resulting coloring we construct an explicit RNA sequence and prove that every different compatible fold has fewer base pairs than the target.  The complete theorem and its supporting lemmas were formalized in Lean~4 and rebuilt from frozen source in an isolated environment.  The work was developed with foundational generative-AI assistance under the author's direction and has not yet received independent human expert review.
\end{abstract}


\noindent\textbf{Subject classifications:} arXiv \texttt{q-bio.QM} (primary), \texttt{q-bio.BM} (cross-list); 2020 Mathematics Subject Classification: 92C40 (primary), 05C05, 05C15, 68W05, 68V20; ACM Computing Classification System: Applied computing---Molecular structural biology; Theory of computation---Design and analysis of algorithms; Theory of computation---Automated reasoning.

\section{Introduction}

\subsection{RNA inverse folding}

RNA inverse folding starts from a target secondary structure and asks for a nucleotide sequence that folds uniquely into that target under a specified energy model.  This reverses the usual structure-prediction problem: the structure is fixed, while the sequence must be designed.  Representative approaches include ViennaRNA's RNAinverse, a canonical early thermodynamic inverse-folding heuristic; the stochastic local search of RNA-SSD; INFO-RNA's dynamic-programming initialization followed by local search; the constraint-programming method RNAiFold; and NUPACK's ensemble-defect optimization \citep{hofacker1994,andronescu2004,busch2006,garciamartin2013,zadeh2011}.  Bonnet et al. proved NP-completeness for a Watson--Crick design-extension formulation with prescribed nucleotide constraints \citep{bonnet2020}, so exact positive results for well-defined structural classes remain useful.

We study the combinatorial Watson--Crick base-pair model studied by Hale\v{s} et al. \citep{hales2017}.  Every permitted base pair contributes the same energy, and loops carry no additional energy term.  A sequence is a design for a target when the target is its unique compatible noncrossing structure with the maximum number of base pairs.  This is an idealized model, not a replacement for Turner nearest-neighbor thermodynamics \citep{turner2010}.  Its advantage is that structural conditions can be proved exactly and formalized completely.

\subsection{From RNA loops to the Hale\v{s} coloring certificate}

The coloring method is easiest to understand by moving back and forth between the RNA drawing familiar to chemists and biologists and the tree used in the proof.  Consecutive nested base pairs form the rungs of a stem or helix.  At the inward end of a stem, a hairpin closes the backbone; a bulge or internal loop leads into one further stem with one or more unpaired nucleotides on the intervening backbone; and a multiloop is a junction at which two or more further stems leave the enclosing stem.  Hale\v{s} et al. encode exactly the same nesting information by a rooted ordered interval tree \citep{hales2017}.  A target pair $(i,j)$ becomes a node $[i,j]$, and each target-unpaired nucleotide $k$ becomes its own singleton leaf $[k,k]$.  Immediately nested pairs are parent and child, so every stem remains a chain with one tree node per base pair; it is not collapsed to a single node.  A virtual root represents the exterior of the molecule, and sibling order records the $5'$-to-$3'$ backbone order.  The tree is therefore the nesting skeleton of the stems and loops, not a different molecular model.

The labels black $\Bcol$, white $\Wcol$, and gray $\Gray$ are sequence-design instructions on target base pairs, not physical colors or energetic states.  For an ordered target pair $(i,j)$ with $i<j$, they prescribe
\[
\Bcol\longmapsto(\G,\C),\qquad
\Wcol\longmapsto(\C,\G),\qquad
\Gray\longmapsto(\A,\U)\ \text{or}\ (\U,\A),
\]
while every target-unpaired position receives $\A$.  Black and white are the two backbone orientations of a G--C pair.  Gray is the A--U class; its orientation is selected later so that the nucleotide identities presented around each loop are locally distinct.

Each loop sees one incidence from every bordering stem.  In the tree, an internal loop node sees the colors of its outgoing paired children and the inward-facing side of its own closing pair.  Viewing that closing pair from the loop reverses its orientation, so its contribution is $\inv(\Bcol)=\Wcol$, $\inv(\Wcol)=\Bcol$, or $\inv(\Gray)=\Gray$.  Hale\v{s} et al.'s \emph{properness} conditions are equivalently the capacity rule
\[
X(v)=\{\inv(c(v))\}\uplus\{c(u):u\text{ is a paired child of }v\}
\quad\text{contains at most one $\Bcol$, one $\Wcol$, and two $\Gray$'s.}
\]
At the exterior root there is no closing incidence.  In the original parent--child formulation, a node has at most one black child, one white child, and two gray children; at most one child shares the parent's color; and black and white are never adjacent as parent and child.  The incidence formulation makes the loop interpretation visible.  It allocates four distinguishable sequence ``ports''---G, C, A, and U---around a loop, and is not a claim about physical loop energetics.

Hale\v{s} et al. make the local rule constructive with a greedy top-down coloring.  At a black parent, color the first paired child black and all remaining paired children gray; at a white parent, use white and then gray; and at a gray parent or the root, use black for the first paired child, white for the second, and gray for the rest.  ``First'' and ``second'' refer only to backbone order.  Under the paired-degree bound this coloring is proper.  The obstruction $m_5$ is the clearest failure of the capacity count: five paired incidences cannot fit into the four slots $\Bcol,\Wcol,\Gray,\Gray$.  The companion obstruction $m_{3\bullet}$ occurs when an unpaired nucleotide is exposed together with too many paired incidences.  Properness controls the target-paired positions locally, but it does not yet stop a distant target-unpaired A from pairing with a U assigned to a gray target pair.

That global protection is supplied by \emph{levels}.  Give $\Bcol$, $\Wcol$, and $\Gray$ the increments $+1$, $-1$, and $0$, respectively, and sum them on the path from the root to a node.  An unpaired singleton inherits the level of its enclosing loop.  Thus the level tracks the relevant G--C prefix imbalance encoded by the orientations along the nesting path.  A proper coloring is \emph{separated} when no integer level is occupied both by a gray paired node and by a target-unpaired node.

The even--odd device used later in this paper is already present in Hale\v{s} et al.'s 2017 construction.  In the proof of their Theorem~10, selected pairs are inflated so that gray paired nodes and unpaired nodes occupy opposite parity classes, after which one type is kept even and the other odd \citep{hales2017}.  Thus parity separation predates the later modulo-$m$ terminology.  Boury et al.'s contribution was to formalize the general modulo-$m$ condition, give a finite-state decision and construction algorithm, and obtain a target-preserving all-long guarantee.

Figure~\ref{fig:hales-bridge} follows one small RNA from its biological drawing to this mathematical certificate.  The incoming stem has color word $\Bcol\Bcol\Bcol$, the left isolated stack has $\Gray\Bcol$, and the right stem has $\Bcol\Bcol\Bcol$.  At the multiloop, the closing black pair is seen as white, so the exposure is $\{\Wcol,\Gray,\Bcol\}$.  The only gray pair has level $3$, whereas the unpaired A has level $4$.

\begin{figure}[t]
\centering
\begin{subfigure}[t]{0.48\linewidth}
\centering
\begin{tikzpicture}[x=1cm,y=0.92cm,font=\scriptsize,
  backbone/.style={DeepBlue,very thick,line cap=round},
  bond/.style={ProofGray,thick},
  call/.style={rounded corners,draw=ProofGray,fill=SoftGray,inner sep=2.5pt,align=center}]
  \draw[backbone] (-0.34,-2.45)--(-0.34,0.06);
  \draw[backbone] (0.34,-2.45)--(0.34,0.06);
  \foreach \y in {-2.05,-1.48,-0.91}{\draw[bond] (-0.34,\y)--(0.34,\y);}
  \draw[backbone] (0,0.78) circle[radius=0.76];
  \begin{scope}[shift={(-0.50,1.34)},rotate=43]
    \draw[backbone] (-0.24,0)--(-0.24,1.42);
    \draw[backbone] (0.24,0)--(0.24,1.42);
    \foreach \y in {0.42,0.94}{\draw[bond] (-0.24,\y)--(0.24,\y);}
    \draw[backbone] (-0.24,1.42) .. controls (-0.82,1.70) and (-0.58,2.22) .. (0,2.28)
      .. controls (0.58,2.22) and (0.82,1.70) .. (0.24,1.42);
    \node[circle,fill=NucA!14,draw=NucA,text=NucA,inner sep=1.4pt,font=\scriptsize\bfseries] at (0,2.26) {A};
  \end{scope}
  \begin{scope}[shift={(0.50,1.34)},rotate=-43]
    \draw[backbone] (-0.24,0)--(-0.24,1.72);
    \draw[backbone] (0.24,0)--(0.24,1.72);
    \foreach \y in {0.38,0.86,1.34}{\draw[bond] (-0.24,\y)--(0.24,\y);}
    \draw[backbone] (-0.24,1.72) .. controls (-0.70,1.96) and (-0.52,2.25) .. (0,2.30)
      .. controls (0.52,2.25) and (0.70,1.96) .. (0.24,1.72);
  \end{scope}
  \node[font=\bfseries] at (0,0.78) {multiloop};
  \node[call] at (-2.15,2.48) {hairpin loop\\one unpaired A};
  \draw[-{Latex[length=1.7mm]},ProofGray] (-1.70,2.30)--(-1.28,1.86);
  \node[call] at (2.06,2.48) {hairpin loop};
  \draw[-{Latex[length=1.7mm]},ProofGray] (1.72,2.30)--(1.34,1.86);
  \node[align=center] at (-1.72,0.30) {isolated stack\\$h=2$};
  \node[align=center] at (1.72,0.30) {stem\\$h=3$};
  \node[align=center] at (1.14,-1.54) {incoming stem\\$h=3$};
  \node[anchor=east] at (-0.43,-2.43) {$5'$};
  \node[anchor=west] at (0.43,-2.43) {$3'$};
\end{tikzpicture}
\caption{A secondary-structure drawing in the paper's $\theta=0$ model.}
\end{subfigure}\hfill
\begin{subfigure}[t]{0.48\linewidth}
\centering
\begin{tikzpicture}[x=1cm,y=0.78cm,font=\scriptsize,
  edge/.style={DeepBlue,thick},
  pcB/.style={circle,draw=black,fill=black,text=white,minimum size=6.5mm,inner sep=0pt},
  pcG/.style={circle,draw=DeepBlue,fill=ProofGray!45,text=black,minimum size=6.5mm,inner sep=0pt},
  root/.style={rounded corners,draw=DeepBlue,fill=SoftBlue,inner sep=2.3pt},
  unp/.style={circle,draw=NucA,fill=NucA!12,text=NucA,minimum size=6.5mm,inner sep=0pt}]
  \node[root] (r) at (0,3.05) {virtual root $\ell=0$};
  \node[pcB] (b1) at (0,2.30) {$\Bcol$};
  \node[pcB] (b2) at (0,1.55) {$\Bcol$};
  \node[pcB] (b3) at (0,0.80) {$\Bcol$};
  \node[pcG] (g1) at (-1.25,-0.05) {$\Gray$};
  \node[pcB] (lb) at (-1.25,-0.85) {$\Bcol$};
  \node[unp] (ua) at (-1.25,-1.70) {A};
  \node[pcB] (rb1) at (1.25,-0.05) {$\Bcol$};
  \node[pcB] (rb2) at (1.25,-0.85) {$\Bcol$};
  \node[pcB] (rb3) at (1.25,-1.65) {$\Bcol$};
  \foreach \a/\b in {r/b1,b1/b2,b2/b3,b3/g1,g1/lb,lb/ua,b3/rb1,rb1/rb2,rb2/rb3}{\draw[edge] (\a)--(\b);}
  \node[anchor=west] at (0.26,2.30) {$\ell=1$};
  \node[anchor=west] at (0.26,1.55) {$\ell=2$};
  \node[anchor=west] at (0.26,0.80) {$\ell=3$; loop $M$};
  \node[anchor=west] at (-0.90,-0.05) {$\ell=3$};
  \node[anchor=west] at (-0.90,-0.85) {$\ell=4$};
  \node[anchor=west,text=NucA] at (-0.90,-1.70) {$\ell=4$};
  \node[anchor=west] at (1.57,-0.05) {$\ell=4$};
  \node[anchor=west] at (1.57,-0.85) {$\ell=5$};
  \node[anchor=west] at (1.57,-1.65) {$\ell=6$};
  \node[rounded corners,draw=GoodGreen,fill=GoodGreen!6,align=center,inner sep=3pt] at (0,-2.65)
    {gray-pair levels $=\{3\}$\\unpaired levels $=\{4\}$: separated};
\end{tikzpicture}
\caption{The same fold as an interval tree.}
\end{subfigure}

\vspace{2mm}
\begin{subfigure}[t]{0.98\linewidth}
\centering
\begin{tikzpicture}[x=1cm,y=1cm,font=\scriptsize,
  pcB/.style={circle,draw=black,fill=black,text=white,minimum size=6.5mm,inner sep=0pt},
  pcW/.style={circle,draw=black,fill=white,text=black,minimum size=6.5mm,inner sep=0pt},
  pcG/.style={circle,draw=DeepBlue,fill=ProofGray!45,text=black,minimum size=6.5mm,inner sep=0pt},
  port/.style={rounded corners,draw=ProofGray,fill=SoftGray,inner sep=3pt,align=center}]
  \node[circle,draw=DeepBlue,fill=SoftBlue,minimum size=10mm,font=\bfseries] (m) at (-4.65,0) {$M$};
  \node[pcW] (pw) at (-4.65,-1.00) {$\Wcol$};
  \node[pcG] (pg) at (-5.72,0.72) {$\Gray$};
  \node[pcB] (pb) at (-3.58,0.72) {$\Bcol$};
  \draw[DeepBlue,thick] (m)--(pw) (m)--(pg) (m)--(pb);
  \node[anchor=west] at (-4.30,-1.00) {closing $\Bcol$ is seen as $\Wcol$; port C};
  \node[anchor=east] at (-6.05,0.72) {A--U port: A};
  \node[anchor=west] at (-3.25,0.72) {G--C port: G};
  \node[port] at (-4.65,-1.72) {$X(M)=\{\Wcol,\Gray,\Bcol\}$ is proper;\\the loop-facing nucleotide ports C, A, G are distinct.};
  \draw[ProofGray] (-1.15,-1.80)--(-1.15,1.18);
  \node[anchor=west,align=left] at (-0.75,0.62) {\textbf{pair-color instructions}\\[2pt]
    $\Bcol\mapsto\textcolor{NucG}{\G}\! -\!\textcolor{NucC}{\C}$,\quad $\delta(\Bcol)=+1$\\
    $\Wcol\mapsto\textcolor{NucC}{\C}\! -\!\textcolor{NucG}{\G}$,\quad $\delta(\Wcol)=-1$\\
    $\Gray\mapsto\textcolor{NucA}{\A}\! -\!\textcolor{NucU}{\U}$ or $\textcolor{NucU}{\U}\! -\!\textcolor{NucA}{\A}$,\quad $\delta(\Gray)=0$\\
    target-unpaired $\mapsto\textcolor{NucA}{\A}$};
\end{tikzpicture}
\caption{The loop-facing capacity rule and its nucleotide meaning.}
\end{subfigure}
\caption{One RNA target in biological and mathematical representations.  (a) The backbone drawing shows stems, hairpins, a multiloop, and an unpaired nucleotide.  (b) Every target base pair becomes one colored tree node, every unpaired nucleotide becomes one leaf, and levels are root-path sums.  (c) At the multiloop, the incoming black closing pair contributes its inverse white incidence.  All pair markers carry explicit $\Bcol$, $\Wcol$, or $\Gray$ labels, so the figure does not rely on hue.}
\label{fig:hales-bridge}
\end{figure}
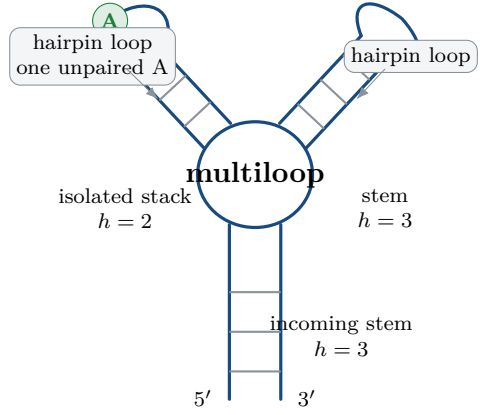
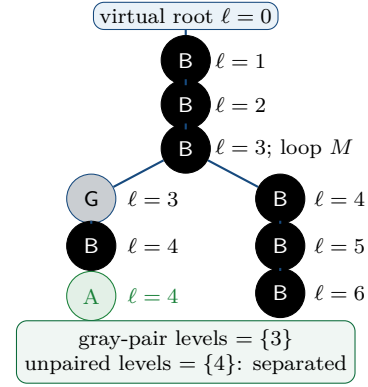
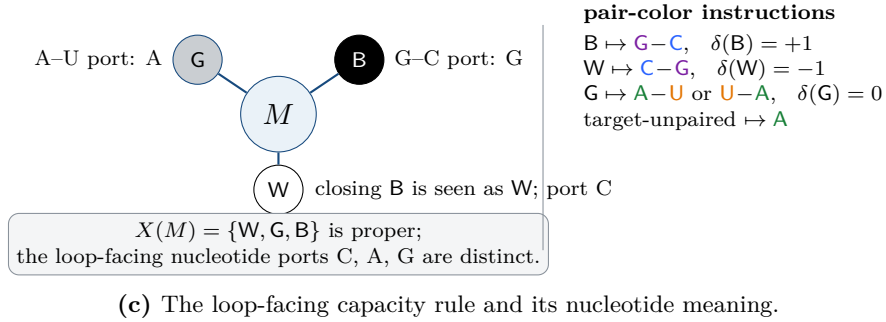

The two parts of the certificate now have different jobs.  Properness makes the target pairing unique after the target-unpaired positions are removed: the exposed nucleotide ports prevent an equal-pair-count local rerouting among the paired positions.  Separation protects the unpaired A's.  If a competing noncrossing fold pairs one of them to a U taken from a gray target pair, the unequal G--C prefix balances at the two endpoints force at least one G or C to remain unused.  The target already pairs every C and every U, and every Watson--Crick pair consumes one C or one U.  The competitor therefore loses a pair and cannot tie the target.  Hale\v{s} et al. proved that a proper separated coloring consequently yields a design in the four-letter Watson--Crick model \citep{hales2017}; Sections~\ref{sec:sequence}--\ref{sec:notie} give the complete argument used here.

This result is a sufficient certificate, not a characterization of designability and not a promise that the greedy coloring of every motif-free target is separated.  In Hale\v{s} et al.'s general structure-approximation construction, a parity conflict in the greedy coloring is repaired by inserting or duplicating at most one base pair in the affected helix and adjusting the coloring, thereby shifting descendant levels \citep{hales2017}.  That proves designability of a modified, enlarged structure; it does not design the original target unchanged.  This distinction is central in inverse folding, where the target is prescribed.

\subsection{Boury's finite-state construction and the length-2 gap}

Building on Hale\v{s} et al.'s parity idea, Boury, Bulteau, and Ponty introduced modulo-$m$ separability as a finite-state, target-preserving method \citep{boury2024,boury2025}.  In a modulo-$m$ separated coloring, gray-pair levels and unpaired-nucleotide levels occupy disjoint residue classes modulo $m$.  This condition is stronger than ordinary separation, but it replaces unbounded integer levels by finitely many states.  Their dynamic program decides and constructs such a coloring, when one exists, in $O(n\,2^m)$ time; in particular, fixed $m=2$ is linear-time.  For $m=2$, black and white both change parity because $+1\equiv-1\equiv1\pmod 2$, while gray preserves parity.  We write $\xi$ for the residue used by unpaired nodes and $\eta=1-\xi$ for the residue used by gray pairs.

Biologically, their recursion contracts each maximal stem to a connection between loops.  Mathematically, that connection is a two-state transfer device.  A hairpin, bulge, or internal loop containing unpaired nucleotides must sit in the $\xi$ state; a multiloop in the construction sits in the $\eta$ state.  The first pair of a helix must present a color allowed by the loop above it, the colors inside the helix must keep every gray pair at residue $\eta$, and the last pair must deliver the residue required by the loop below it.  Boury et al. proved that a helix of at least three pairs has enough internal positions to retain a prescribed first color and realize either required exit residue.  Consequently every $(m_5,m_{3\bullet})$-free target whose helices all have length at least three admits a modulo-2 separated coloring and a linear-time design \citep{boury2024,boury2025}.  The third pair has no special thermodynamic significance here; it is an additional combinatorial switching position in the maximum-pair model.

The non-gray-ending transfer words selected below are already displayed in Boury et al.'s WABI Figure~8 and journal Figure~12 \citep{boury2024,boury2025}.  Up to black--white symmetry, those figures tabulate the families $a^h$, $a\Gray a^{h-2}$, $aa\Gray a^{h-3}$, $\Gray a^{h-1}$, and $\Gray\Gray a^{h-2}$.  Lemma~\ref{lem:strong-long} uses all but the $aa\Gray a^{h-3}$ family.  The selected words preserve the prescribed first color, place gray pairs in the required residue, realize the needed exit residue, and visibly end non-gray.  The lemma records the exact subfamily needed by the global induction; it does not introduce new local color words.

An isolated stack---a maximal helix of exactly two pairs---has only one position after its head and therefore no comparable spare switch.  The complete transition table, proved as Lemma~\ref{lem:two-pair}, is
\[
\begin{array}{c@{\quad\longrightarrow\quad}c@{\qquad}l}
\text{entry residue} & \text{exit residue} & \text{valid two-pair color words}\\[1mm]
\xi & \xi & \Bcol\Bcol,\ \Wcol\Wcol,\\
\xi & \eta & \Bcol\Gray,\ \Wcol\Gray,\\
\eta & \xi & \Gray\Bcol,\ \Gray\Wcol,\\
\eta & \eta & \Bcol\Bcol,\ \Wcol\Wcol,\ \Gray\Gray.
\end{array}
\]
For instance, an $\eta\to\xi$ isolated stack must start gray.  In the $\eta\to\eta$ case, if the loop above also prescribes a gray head, the only two-pair word is $\Gray\Gray$, which necessarily closes gray.  Thus two pairs cannot in general satisfy the entry color, the gray-residue restriction, the exit residue, and the terminal color independently.  This is exactly where the all-long guarantee stops; it is not evidence that targets containing isolated stacks are undesignable.

The missing terminal-color control becomes visible when branches meet.  Suppose a multiloop at residue $\eta$ has two outgoing child subtrees containing isolated stacks.  The weaker short-helix guarantee may require both child heads to be gray.  If the ordinary all-long transfer above the loop is reused with a gray-ending $\eta$ exit, the loop exposes the inverse of that gray close together with the two gray heads: $\{\Gray,\Gray,\Gray\}$.  Properness permits only two gray incidences.

The global bound of two isolated stacks supplies exactly the missing position.  Once the two outgoing child subtrees contain both isolated stacks allowed in the target, the incoming helix cannot also have length 2.  Because length-1 helices are excluded, it has length at least three.  We can therefore select a non-gray-ending $\eta\to\eta$ word from Boury et al.'s transfer table.  At length three the published family specializes to $\Gray\Bcol\Bcol$ or, by black--white symmetry, $\Gray\Wcol\Wcol$.  These choices change the loop exposures to $\{\Wcol,\Gray,\Gray\}$ or $\{\Bcol,\Gray,\Gray\}$, both proper.  Figure~\ref{fig:boury-gap} shows the bottleneck and its bounded-resource resolution on RNA-shaped stem--loop drawings.

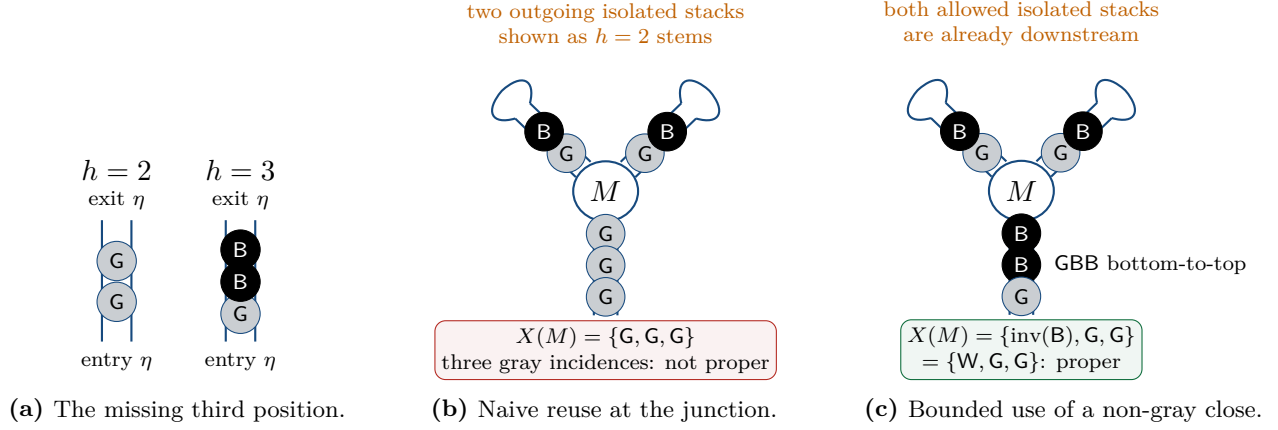
\begin{figure}[t]
\centering
\begin{subfigure}[t]{0.29\linewidth}
\centering
\begin{tikzpicture}[x=0.78cm,y=0.72cm,font=\scriptsize,
  rail/.style={DeepBlue,thick}, bond/.style={ProofGray,thick},
  pcB/.style={circle,draw=black,fill=black,text=white,minimum size=5.2mm,inner sep=0pt},
  pcG/.style={circle,draw=DeepBlue,fill=ProofGray!45,text=black,minimum size=5.2mm,inner sep=0pt}]
  \node[font=\bfseries] at (-1.05,3.10) {$h=2$};
  \draw[rail] (-1.30,-0.05)--(-1.30,2.18);
  \draw[rail] (-0.80,-0.05)--(-0.80,2.18);
  \foreach \y in {0.68,1.43}{\draw[bond] (-1.30,\y)--(-0.80,\y);}
  \node[pcG] at (-1.05,0.68) {$\Gray$};
  \node[pcG] at (-1.05,1.43) {$\Gray$};
  \node at (-1.05,-0.42) {entry $\eta$};
  \node at (-1.05,2.55) {exit $\eta$};

  \node[font=\bfseries] at (1.05,3.10) {$h=3$};
  \draw[rail] (0.80,-0.05)--(0.80,2.18);
  \draw[rail] (1.30,-0.05)--(1.30,2.18);
  \foreach \y in {0.46,1.05,1.64}{\draw[bond] (0.80,\y)--(1.30,\y);}
  \node[pcG] at (1.05,0.46) {$\Gray$};
  \node[pcB] at (1.05,1.05) {$\Bcol$};
  \node[pcB] at (1.05,1.64) {$\Bcol$};
  \node at (1.05,-0.42) {entry $\eta$};
  \node at (1.05,2.55) {exit $\eta$};
\end{tikzpicture}
\caption{The missing third position.}
\end{subfigure}\hfill
\begin{subfigure}[t]{0.34\linewidth}
\centering
\begin{tikzpicture}[x=0.78cm,y=0.72cm,font=\scriptsize,
  rail/.style={DeepBlue,thick}, bond/.style={ProofGray,thick},
  pcB/.style={circle,draw=black,fill=black,text=white,minimum size=5.1mm,inner sep=0pt},
  pcG/.style={circle,draw=DeepBlue,fill=ProofGray!45,text=black,minimum size=5.1mm,inner sep=0pt},
  bad/.style={rounded corners,draw=AlertRed,fill=AlertRed!6,align=center,inner sep=2.5pt}]
  \draw[rail] (0,0) circle[radius=0.55];
  \node[font=\bfseries] at (0,0) {$M$};
  \draw[rail] (-0.20,-0.50)--(-0.20,-2.28);
  \draw[rail] (0.20,-0.50)--(0.20,-2.28);
  \foreach \y in {-0.78,-1.36,-1.94}{\draw[bond] (-0.20,\y)--(0.20,\y);\node[pcG] at (0,\y) {$\Gray$};}
  \begin{scope}[shift={(-0.38,0.40)},rotate=46]
    \draw[rail] (-0.18,0)--(-0.18,1.55); \draw[rail] (0.18,0)--(0.18,1.55);
    \foreach \y in {0.43,1.00}{\draw[bond] (-0.18,\y)--(0.18,\y);}
    \node[pcG] at (0,0.43) {$\Gray$};
    \node[pcB] at (0,1.00) {$\Bcol$};
    \draw[rail] (-0.18,1.55) .. controls (-0.55,1.78) and (-0.45,2.05) .. (0,2.08) .. controls (0.45,2.05) and (0.55,1.78) .. (0.18,1.55);
  \end{scope}
  \begin{scope}[shift={(0.38,0.40)},rotate=-46]
    \draw[rail] (-0.18,0)--(-0.18,1.55); \draw[rail] (0.18,0)--(0.18,1.55);
    \foreach \y in {0.43,1.00}{\draw[bond] (-0.18,\y)--(0.18,\y);}
    \node[pcG] at (0,0.43) {$\Gray$};
    \node[pcB] at (0,1.00) {$\Bcol$};
    \draw[rail] (-0.18,1.55) .. controls (-0.55,1.78) and (-0.45,2.05) .. (0,2.08) .. controls (0.45,2.05) and (0.55,1.78) .. (0.18,1.55);
  \end{scope}
  \node[align=center,text=orange!75!black] at (0,3.08) {two outgoing isolated stacks\\shown as $h=2$ stems};
  \node[bad] at (0,-2.92) {$X(M)=\{\Gray,\Gray,\Gray\}$\\three gray incidences: not proper};
\end{tikzpicture}
\caption{Naive reuse at the junction.}
\end{subfigure}\hfill
\begin{subfigure}[t]{0.34\linewidth}
\centering
\begin{tikzpicture}[x=0.78cm,y=0.72cm,font=\scriptsize,
  rail/.style={DeepBlue,thick}, bond/.style={ProofGray,thick},
  pcB/.style={circle,draw=black,fill=black,text=white,minimum size=5.1mm,inner sep=0pt},
  pcG/.style={circle,draw=DeepBlue,fill=ProofGray!45,text=black,minimum size=5.1mm,inner sep=0pt},
  good/.style={rounded corners,draw=GoodGreen,fill=GoodGreen!6,align=center,inner sep=2.5pt}]
  \draw[rail] (0,0) circle[radius=0.55];
  \node[font=\bfseries] at (0,0) {$M$};
  \draw[rail] (-0.20,-0.50)--(-0.20,-2.28);
  \draw[rail] (0.20,-0.50)--(0.20,-2.28);
  \foreach \y in {-0.78,-1.36,-1.94}{\draw[bond] (-0.20,\y)--(0.20,\y);}
  \node[pcB] at (0,-0.78) {$\Bcol$};
  \node[pcB] at (0,-1.36) {$\Bcol$};
  \node[pcG] at (0,-1.94) {$\Gray$};
  \node[anchor=west] at (0.38,-1.36) {$\Gray\Bcol\Bcol$ bottom-to-top};
  \begin{scope}[shift={(-0.38,0.40)},rotate=46]
    \draw[rail] (-0.18,0)--(-0.18,1.55); \draw[rail] (0.18,0)--(0.18,1.55);
    \foreach \y in {0.43,1.00}{\draw[bond] (-0.18,\y)--(0.18,\y);}
    \node[pcG] at (0,0.43) {$\Gray$};
    \node[pcB] at (0,1.00) {$\Bcol$};
    \draw[rail] (-0.18,1.55) .. controls (-0.55,1.78) and (-0.45,2.05) .. (0,2.08) .. controls (0.45,2.05) and (0.55,1.78) .. (0.18,1.55);
  \end{scope}
  \begin{scope}[shift={(0.38,0.40)},rotate=-46]
    \draw[rail] (-0.18,0)--(-0.18,1.55); \draw[rail] (0.18,0)--(0.18,1.55);
    \foreach \y in {0.43,1.00}{\draw[bond] (-0.18,\y)--(0.18,\y);}
    \node[pcG] at (0,0.43) {$\Gray$};
    \node[pcB] at (0,1.00) {$\Bcol$};
    \draw[rail] (-0.18,1.55) .. controls (-0.55,1.78) and (-0.45,2.05) .. (0,2.08) .. controls (0.45,2.05) and (0.55,1.78) .. (0.18,1.55);
  \end{scope}
  \node[align=center,text=orange!75!black] at (0,3.08) {both allowed isolated stacks\\are already downstream};
  \node[good] at (0,-2.92) {$X(M)=\{\inv(\Bcol),\Gray,\Gray\}$\\$=\{\Wcol,\Gray,\Gray\}$: proper};
\end{tikzpicture}
\caption{Bounded use of a non-gray close.}
\end{subfigure}
\caption{Why the all-long construction does not apply unchanged when isolated stacks are admitted.  (a) With a gray head and an $\eta\to\eta$ requirement, a two-pair helix must use $\Gray\Gray$ and close gray; a third pair permits the already published Boury word $\Gray\Bcol\Bcol$ and a non-gray close.  (b) At a multiloop with two gray child heads, a gray close creates three gray incidences.  (c) In the illustrative immediate-stack case drawn here, the two outgoing stems already account for both allowed isolated stacks, so the incoming helix is long and the induction selects a non-gray-ending word from Boury et al.'s transfer table.  In the general theorem, either isolated stack may occur deeper in the corresponding child subtree; the same two gray head demands and counting argument apply.  The failed coloring in (b) is a failure of a naive local choice, not a proof of target non-designability.  Letters on the rungs make the diagrams color-blind accessible \citep{boury2024,boury2025}.}
\label{fig:boury-gap}
\end{figure}

\subsection{Main result and proof outline}

\begin{contributionbox}
The paper makes four main contributions.
\begin{enumerate}
\item \textbf{At-most-two theorem.}  Every $(m_5,m_{3\bullet})$-free target with no isolated base pair, at most two isolated stacks, and all remaining helices of length at least three admits a modulo-2 separated coloring and is designable in the Watson--Crick base-pair model.
\item \textbf{Induction on helix subtrees.}  Each helix subtree is summarized by the entry states it is guaranteed to support.  The sets $F$ and $Q$ distinguish subtrees without an isolated stack from those containing at least one.
\item \textbf{Self-contained design proof.}  We include the color-to-sequence construction, the nucleotide-counting bound, uniqueness on the paired restriction, and the exclusion of alternative optimal folds.
\item \textbf{Formal verification.}  The complete theorem was formalized in Lean~4.  The frozen source was rebuilt in isolation, and a blinded automated AI audit separately compared the formal statement with an anonymous mathematical specification.
\end{enumerate}
\end{contributionbox}

The theorem is a sufficient condition, not a characterization of all designable targets with isolated stacks.  The target $((.))((.))((.))$ admits no proper modulo-2 separated coloring at all.  It nevertheless has a proper ordinarily separated coloring and is uniquely designable, so the example marks a sharp boundary for this modulo-2 criterion rather than a general designability boundary.

\subsection{AI-assisted research and manuscript preparation}
\label{sec:ai-disclosure}

This project was carried out with foundational and pervasive assistance from generative-AI systems---principally OpenAI's ChatGPT and Codex and Anthropic's Claude Code---under the human author's direction.  Their role extended across all major stages: framing and refining the research problem and theorem statement; assisting with literature searches and citation and provenance checks; proposing, revising, and stress-testing mathematical arguments; generating and reviewing Lean code, test programs, computational checkers, and audit procedures; searching for counterexamples and hidden assumptions; developing examples and figures; assembling reproducibility and release artifacts; and drafting, restructuring, editing, proofreading, and typesetting the manuscript.  These systems were foundational tools in both the research and the writing, not merely language editors or after-the-fact verifiers.

The human author directed the work; chose and approved the final scientific model, scope, theorem statement, acceptance criteria, and canonical versions; decided which AI-generated material to accept, revise, or reject; and approved the final manuscript and release.  AI output was not treated as evidence merely because one or more systems agreed.  The author takes full responsibility for the final model, theorem statement, proof exposition, formalization, computations, figures, citations, and all claims.  Several early formulations were rejected or repaired before the final proof was frozen.

The evidence is ranked rather than treated as a vote among AI systems.  In descending order of logical strength, it consists of Lean kernel acceptance of the exact final theorem; isolated clean reconstruction of the frozen source and pinned dependencies; source-integrity, dependency, and axiom audits; a blinded automated source-to-specification review by a fresh AI session; and independent finite computational checks and negative controls.  The correctness claim does not rest on agreement among AI systems.  Its strongest evidence is kernel checking of the formal proof and isolated reconstruction of the frozen source.  The AI reviews are adversarial checks of interpretation, model fidelity, and exposition.  Section~\ref{sec:formal} records these checks and their limitations.

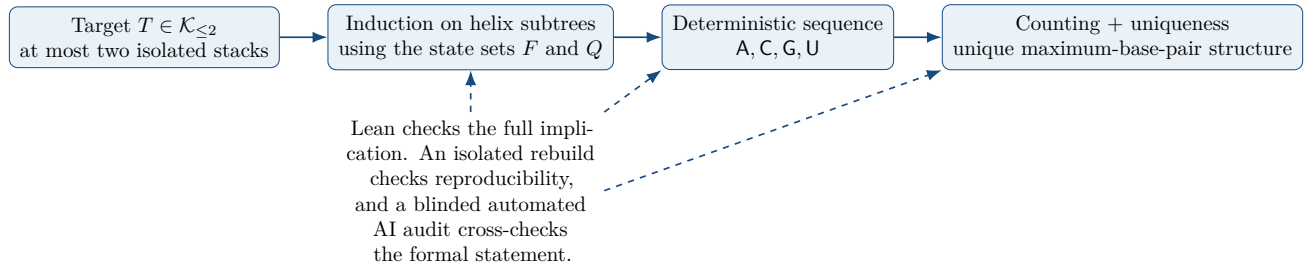
\begin{figure}[t]
\centering
\resizebox{\linewidth}{!}{%
\begin{tikzpicture}[node distance=8mm and 8mm, every node/.style={font=\small},
box/.style={draw=DeepBlue,rounded corners,fill=SoftBlue,align=center,minimum height=9mm,inner sep=5pt},
arr/.style={-{Latex[length=2.5mm]},thick,DeepBlue}]
\node[box] (target) {Target $T\in\Ktwo$\\at most two isolated stacks};
\node[box,right=of target] (states) {Induction on helix subtrees\\using the state sets $F$ and $Q$};
\node[box,right=of states] (seq) {Deterministic sequence\\$\A,\C,\G,\U$};
\node[box,right=of seq] (uniq) {Counting + uniqueness\\unique maximum-base-pair structure};
\draw[arr] (target)--(states);
\draw[arr] (states)--(seq);
\draw[arr] (seq)--(uniq);
\node[below=7mm of states,align=center,text width=4.2cm] (formal) {Lean checks the full implication.  An isolated rebuild checks reproducibility, and a blinded automated AI audit cross-checks the formal statement.};
\draw[arr,dashed] (formal.north)--(states.south);
\draw[arr,dashed] (formal.north east)--(seq.south west);
\draw[arr,dashed] (formal.east)--(uniq.south west);
\end{tikzpicture}%
}
\caption{Structure of the proof.  The new step is the global resource-counting induction that constructs a proper modulo-2 separated coloring for targets with at most two isolated stacks, using previously published local transfer words.  The remaining implications are proved self-containedly and formalized in Lean.}
\label{fig:pipeline}
\end{figure}

Figure~\ref{fig:pipeline} separates the new coloring construction from the established coloring-to-design implication and from the three layers of formal-artifact checking.

\subsection{Plain-language glossary}
\label{sec:glossary}

The following definitions are reading aids for the terminology used throughout the paper; the exact mathematical definitions in Section~\ref{sec:model} and later sections remain authoritative.  Three distinctions are especially important.  A pair ``color'' is an instruction for assigning nucleotides, not a physical color.  A ``level'' is an integer used for bookkeeping, not an energy, distance, or height.  And a modulo-2 separated coloring still uses the three pair colors $\Bcol$, $\Wcol$, and $\Gray$: ``modulo 2'' refers only to whether a level is even or odd.

{\small
\setlength{\LTpre}{5pt}
\setlength{\LTpost}{5pt}
\renewcommand{\arraystretch}{1.12}
\begin{longtable}{@{}>{\raggedright\arraybackslash\bfseries}p{0.235\linewidth}>{\raggedright\arraybackslash}p{0.705\linewidth}@{}}
\toprule
Term & Plain-language meaning in this paper \\
\midrule
\endfirsthead
\toprule
Term & Plain-language meaning in this paper, continued \\
\midrule
\endhead
\bottomrule
\endfoot
\multicolumn{2}{@{}l}{\textcolor{DeepBlue}{\bfseries RNA structure and the design problem}}\\
\addlinespace[2pt]
Target or target structure & The desired pattern of paired and unpaired nucleotide positions.  It is fixed before the sequence is chosen.  A secondary-structure target records this pairing pattern, not the molecule's full three-dimensional shape. \\
\addlinespace[3pt]
Compatible fold & A fold that the chosen sequence can form in this model: every one of its pairs is A--U, U--A, C--G, or G--C.  ``Compatible'' means allowed, not necessarily optimal or unique. \\
\addlinespace[3pt]
Design; designable & A sequence is a \emph{design} when the target is its only compatible noncrossing fold with the largest possible number of pairs.  A target is \emph{designable} if at least one such sequence exists. \\
\addlinespace[3pt]
Maximum-pair model & The simplified energy model used here: each allowed Watson--Crick pair receives the same score, so more pairs means lower energy.  It omits wobble pairs, stacking energies, loop entropies, salt effects, kinetics, and three-dimensional interactions. \\
\addlinespace[3pt]
Noncrossing; pseudoknot-free & If base pairs are drawn as arcs above the sequence, their arcs may be nested or disjoint but may not cross.  This is the usual pseudoknot-free secondary-structure condition. \\
\addlinespace[3pt]
Stem or helix & A maximal uninterrupted ladder of consecutively nested target base pairs.  Its length is the number of paired rungs.  ``Maximal'' matters: two rungs lying inside a longer stem do not form a separate length-2 helix. \\
\addlinespace[3pt]
Hairpin, bulge/internal loop, multiloop & A hairpin caps a stem without another stem leaving it.  A bulge or internal loop contains unpaired nucleotides and continues into one stem.  A multiloop is a junction with two or more outgoing stems.  In the proof, $E$, $L$, and $M$ abbreviate an empty terminus, a loop with unpaired nucleotides and at most one outgoing helix, and an unpaired-free multiloop with two or three outgoing helices, respectively. \\
\addlinespace[3pt]
Isolated base pair; isolated stack & A maximal helix of length 1 is an isolated base pair; a maximal helix of exactly length 2 is an isolated stack.  ``Isolated stack'' does not mean any adjacent pair of rungs chosen from a longer helix. \\

\midrule
\multicolumn{2}{@{}l}{\textcolor{DeepBlue}{\bfseries The coloring certificate}}\\
\addlinespace[2pt]
Interval tree; virtual root & A nesting map of the target.  Every target base pair becomes one paired node, every target-unpaired nucleotide becomes one leaf, and a virtual root represents the exterior.  A stem therefore appears as a chain of nodes rather than one collapsed node. \\
\addlinespace[3pt]
Pair colors $\Bcol$, $\Wcol$, $\Gray$ & Sequence-design labels: $\Bcol$ prescribes G--C, $\Wcol$ prescribes C--G, and gray $\Gray$ prescribes A--U or U--A.  Every target-unpaired position receives A.  The symbol $\Gray$ for the gray pair color is not the nucleotide guanine $\G$. \\
\addlinespace[3pt]
Inverse color; exposure $X(v)$ & A loop sees the first colors of its outgoing stems plus the loop-facing side of its own closing pair.  Viewing a closing pair from the other side exchanges $\Bcol$ with $\Wcol$ and leaves $\Gray$ unchanged.  The resulting \emph{exposure} is a multiset, so repeated colors count separately.  The exterior root has no closing contribution. \\
\addlinespace[3pt]
Proper coloring & A local capacity rule: every exposure contains at most one $\Bcol$, one $\Wcol$, and two gray $\Gray$ incidences.  This permits distinct nucleotide identities to be presented around each loop.  It is not a statement about physical loop energetics, and by itself it does not yet guarantee a design. \\
\addlinespace[3pt]
Increment; level & The colors contribute numerical increments $+1$, $-1$, and $0$ for $\Bcol$, $\Wcol$, and $\Gray$.  A paired node's \emph{level} is the running sum from the virtual root through that node, including its own color.  An unpaired nucleotide inherits the level of its enclosing loop, or level $0$ at the exterior.  Level is not tree depth, energy, helix length, or physical height. \\
\addlinespace[3pt]
Separation; separated coloring & \emph{Separation} means that no gray target pair and no target-unpaired nucleotide have the same exact integer level.  Gray pairs may share levels with other gray pairs, and unpaired positions may share levels with other unpaired positions; the two classes are what must be separated.  A \emph{separated coloring} is a proper coloring with this global property.  It is a sufficient design certificate, not a necessary condition for every designable target. \\
\addlinespace[3pt]
Residue; $\xi$ and $\eta$ & Here \emph{residue} means an arithmetic remainder, not a nucleotide residue.  Modulo two, a level has residue $0$ or $1$---equivalently, it is even or odd.  In the uniform form constructed here, $\xi$ names the residue of every unpaired node and $\eta=1-\xi$ names the other residue, occupied by every gray pair.  The symbols are roles, not fixed choices of even and odd. \\
\addlinespace[3pt]
Modulo-$m$ separation & Instead of comparing complete integer levels, compare only their remainders after division by $m$.  No remainder used by a gray-pair level may be used by an unpaired-node level.  This is stronger than ordinary separation because disjoint remainder classes cannot contain an equal integer level. \\
\addlinespace[3pt]
Modulo-2 separated coloring & Modulo-$m$ separation with $m=2$, so only even versus odd is remembered.  In the nonvacuous uniform form used by this theorem, all unpaired nodes have parity $\xi$ and all gray pairs have the opposite parity $\eta$.  The phrase is sometimes shortened informally to ``modulo-2 coloring,'' but it is not a two-coloring: the pair labels remain $\Bcol$, $\Wcol$, and $\Gray$.  Failure of this certificate does not imply failure of ordinary separation or designability. \\

\midrule
\multicolumn{2}{@{}l}{\textcolor{DeepBlue}{\bfseries Terms used in the recursive construction}}\\
\addlinespace[2pt]
Color word & The ordered list of pair colors along one helix, read from its outermost first pair toward its innermost terminal pair.  For example, $\Gray\Bcol\Bcol$ describes a three-pair helix, not the nucleotide word GBB. \\
\addlinespace[3pt]
Helix head; terminal pair; entry and exit residue & The \emph{head} is the helix's first, outermost pair; the \emph{terminal pair} is its last, innermost pair.  The entry residue is the parent level's parity immediately before the head color is added.  The exit residue is the inclusive parity at the terminal pair.  An \emph{entry state} records both the entry residue and the prescribed head color. \\
\addlinespace[3pt]
Helix transfer & A permitted $\Bcol/\Wcol/\Gray$ color word that carries a specified entry residue and head color to a requested exit residue, while keeping adjacent exposures proper and every gray pair at residue $\eta$.  ``Transfer'' is proof bookkeeping, not a chemical transport process. \\
\addlinespace[3pt]
All-long target & A target in which every maximal helix contains at least three base pairs.  Boury et al.'s earlier guarantee applies to motif-free all-long targets. \\
\addlinespace[3pt]
Stack-bearing subtree & A branch of the interval tree that contains an isolated stack somewhere below it.  Its first outgoing helix need not itself be short; the term describes the whole descendant branch. \\
\addlinespace[3pt]
Feasible-state sets $F$ and $Q$ & Two finite lists of entry states guaranteed by the induction.  $F$ is the guarantee for a helix subtree with no isolated stack; $Q$ is the weaker guarantee for a subtree containing at least one.  They are bookkeeping sets, not additional RNA objects. \\
\addlinespace[3pt]
Motifs $m_5$ and $m_{3\bullet}$ & Two local capacity obstructions.  Paired degree counts the closing paired incidence plus outgoing paired children, or only outgoing children at the exterior root.  Motif $m_5$ means paired degree greater than $4$; $m_{3\bullet}$ means an unpaired child is present while paired degree is greater than $2$. \\
\addlinespace[3pt]
Target class $\Ktwo$ & The targets covered by the theorem: no isolated base pair, at most two isolated stacks, every other helix of length at least three, and neither obstruction motif $m_5$ nor $m_{3\bullet}$. \\
\addlinespace[3pt]
Structure approximation & A method that resolves a coloring conflict by adding or duplicating a base pair in a helix.  It proves a statement about a modified target, not a design of the original target unchanged. \\
\addlinespace[3pt]
Formal verification in Lean & The proof has been translated into definitions and logical steps checked by the Lean kernel.  This checks that the theorem follows from the stated formal assumptions; it does not establish biological realism, experimental usefulness, novelty, or independent human expert review. \\
\end{longtable}
}

\subsection{Organization}

Section~\ref{sec:glossary} provides a reading guide to the terminology.  Sections~\ref{sec:model}--\ref{sec:prior} define the model and summarize prior results.  Sections~\ref{sec:structural}--\ref{sec:globalcolor} construct the coloring.  Sections~\ref{sec:sequence}--\ref{sec:notie} convert it into a unique design.  Section~\ref{sec:examples} gives examples, Section~\ref{sec:formal} describes the Lean development and artifact checks, Section~\ref{sec:ai} records the AI-assisted development and audit trail, Section~\ref{sec:scope} discusses scope and open questions, and Section~\ref{sec:conclusion} concludes.

\section{Model and definitions}
\label{sec:model}

\subsection{Sequences, structures, and energy}

Fix $n\ge 0$ and let $[n]=\{1,\dots,n\}$.  An RNA sequence is a word $w\in\{\A,\C,\G,\U\}^n$.  The Watson--Crick complement involution is
\[
\comp(\A)=\U,
\qquad \comp(\U)=\A,
\qquad \comp(\C)=\G,
\qquad \comp(\G)=\C.
\]
An ordered pair of nucleotides is compatible precisely when the second is the complement of the first.  G--U wobble is not permitted.

A secondary structure $P$ on $[n]$ is a set of arcs $(i,j)$ satisfying:
\begin{enumerate}
\item $1\le i<j\le n$;
\item each position belongs to at most one arc; and
\item the arcs are noncrossing: no two arcs $(i,j),(k,l)$ satisfy $i<k<j<l$.
\end{enumerate}
We impose minimum arc length $\theta=0$, so adjacent positions may pair.  A structure is compatible with $w$ if every arc joins complementary letters.  Its energy is
\[
E(w,P)=-|P|.
\]
Thus lower energy is exactly larger pair count.

\begin{definition}[Design and designability]
A sequence $w$ is a \emph{design} for a target $T$ if $T$ is compatible with $w$ and every distinct compatible noncrossing structure $S$ on the same sequence satisfies
\[
|S|<|T|.
\]
A target is \emph{designable} if it admits a design.
\end{definition}

The quantifier order is important: one sequence is chosen first, and then every compatible noncrossing structure of that same sequence is considered.  Competitors are not required to share the target's stems, helix decomposition, tree, or number of pairs.

\subsection{Tree representation}

Associate to a structure $P$ the rooted ordered tree representation used by Hale\v{s} et al.; we also refer to it as the interval tree.  The nodes are:
\begin{itemize}
\item a virtual root $r=[0,n+1]$;
\item a paired node $[i,j]$ for each arc $(i,j)\in P$; and
\item an unpaired singleton $[k,k]$ for every position not incident to an arc.
\end{itemize}
The children of an interval are its maximal proper subintervals in backbone order.  Noncrossingness makes the family laminar, so every nonroot node has a unique smallest enclosing paired interval or the virtual root.

For a paired node $v$, let $d(v)$ be its number of paired children.  Its paired degree is
\[
\deg(v)=1+d(v)\quad\text{for nonroot paired }v,
\qquad
\deg(r)=d(r).
\]
The root has no parent contribution, and unpaired children do not contribute to paired degree.

\subsection{Forbidden motifs}

The two standard local obstructions are:
\begin{itemize}
\item $m_5$: some node has paired degree greater than $4$;
\item $m_{3\bullet}$: some node has an unpaired child and paired degree greater than $2$.
\end{itemize}

\begin{lemma}[Motif bounds]
If a target avoids $m_5$ and $m_{3\bullet}$, then:
\begin{center}
\begin{tabular}{lll}
\toprule
node & has an unpaired child? & maximum paired children \\
\midrule
nonroot paired & yes & $1$ \\
nonroot paired & no & $3$ \\
root & yes & $2$ \\
root & no & $4$ \\
\bottomrule
\end{tabular}
\end{center}
\end{lemma}
\begin{proof}
For a nonroot paired node, $\deg=1+d$.  If it has an unpaired child, avoiding $m_{3\bullet}$ gives $1+d\le2$; otherwise avoiding $m_5$ gives $1+d\le4$.  The root has degree $d$, giving the remaining rows.
\end{proof}

\subsection{Helices and the target class}

A helix of length $h$ is a maximal consecutive nested run
\[
(i,j),(i+1,j-1),\dots,(i+h-1,j-h+1).
\]
A length-1 helix is an isolated base pair; a length-2 helix is an isolated stack.  Maximality is required both outward and inward.

\begin{lemma}[Helix chains]
Two pairs $[i,j]$ and $[i+1,j-1]$ are consecutive in a helix if and only if $[i,j]$ has exactly one paired child, namely $[i+1,j-1]$, and has no unpaired child.
\end{lemma}
\begin{proof}
If the pairs are stacked, there are no backbone positions available for another child between their corresponding endpoints.  Conversely, if $[i,j]$ has one paired child $[k,l]$ and no unpaired child, its children cover all positions strictly inside $(i,j)$; hence $k=i+1$ and $l=j-1$.
\end{proof}

The consecutive-pair relation partitions the target pairs into disjoint maximal chains, so every target pair belongs to exactly one maximal helix.

\begin{definition}[At-most-two target class]
Let $\Ktwo$ be the class of targets satisfying:
\begin{enumerate}
\item at most two maximal helices have length $2$;
\item no maximal helix has length $1$;
\item every maximal helix whose length is not $2$ has length at least $3$;
\item $m_5$ is absent; and
\item $m_{3\bullet}$ is absent.
\end{enumerate}
The all-unpaired target belongs to $\Ktwo$ vacuously.
\end{definition}

Condition 2 is stated explicitly to keep the scientific boundary visible and to mirror the frozen Lean predicate.  Logically it follows from condition 3, because a helix of length 1 has length different from 2 but cannot have length at least 3.

\begin{remark}[Terminology]
We follow the terminology of Hale\v{s} et al. and Boury et al.: \emph{tree representation}, \emph{proper coloring}, \emph{level}, \emph{separated coloring}, \emph{isolated base pair}, and \emph{isolated stack}.  The symbols $E$, $L$, and $M$ are only local abbreviations for the three terminal-loop types defined below.  The symbols $F$ and $Q$ name two explicitly listed sets of feasible entry states; they are not additional biological objects.
\end{remark}

\subsection{Colors, exposures, and levels}

Color every nonroot paired node black $\Bcol$, white $\Wcol$, or gray $\Gray$.  Define
\[
\inv(\Bcol)=\Wcol,
\qquad \inv(\Wcol)=\Bcol,
\qquad \inv(\Gray)=\Gray,
\]
and
\[
\delta(\Bcol)=+1,
\qquad \delta(\Wcol)=-1,
\qquad \delta(\Gray)=0.
\]
For a nonroot paired node $v$, its exposed multiset is
\[
X(v)=\{\inv(c(v))\}\uplus\{c(u):u\text{ is a paired child of }v\}.
\]
At the root,
\[
X(r)=\{c(u):u\text{ is a paired child of }r\}.
\]

\begin{definition}[Proper coloring]
A coloring is proper if every exposed multiset contains at most one black, at most one white, and at most two gray incidences.
\end{definition}

At the root, properness is exactly the capacity bound on its paired children: at most one black, one white, and two gray.  At a nonroot black node the child capacities are one black, no white, and two gray; at a white node they are no black, one white, and two gray; and at a gray node they are one of each color.  Equivalently, away from the root, a node has at most one child of its own color, black and white are never parent--child neighbors, and the remaining color capacities above also hold.  This is the properness formulation used in the separated-coloring literature \citep{hales2017,boury2024}.

The inclusive level of a paired node is
\[
\level(v)=\sum_{u\in r\leadsto v}\delta(c(u)),
\]
where the sum runs over the colored nonroot paired nodes on the path from the root to $v$, including $v$.  The root level is zero.  An unpaired child of a paired node inherits the parent's level; an unpaired root child has level zero.  The entry level of a nonroot paired node is
\[
\entry(v)=\level(\operatorname{parent}(v))
        =\level(v)-\delta(c(v)).
\]

\begin{definition}[Separated and modulo-2 separated]
A proper coloring is \emph{separated} if no integer level is occupied both by a gray paired node and by an unpaired node.  Following Boury et al., a coloring is \emph{modulo-2 separated} if the residues of gray paired nodes and unpaired nodes are disjoint modulo two.  Our construction yields the more specific situation in which there are residues $\xi,\eta\in\mathbb Z/2\mathbb Z$, with $\eta=1-\xi$, such that every unpaired node has residue $\xi$ and every gray paired node has residue $\eta$.  We call this the \emph{uniform residue form} of modulo-2 separation.  It is the property named \texttt{StrongTwoSeparated} in the Lean development.
\end{definition}

\begin{lemma}
Every coloring in the uniform residue form is separated.
\end{lemma}
\begin{proof}
The gray and unpaired level sets lie in different parity classes.
\end{proof}

\begin{lemma}[Nonvacuous modulo-2 separation has uniform residue form]
\label{lem:modtwo-uniform}
If a proper modulo-2 separated coloring has at least one unpaired node and at least one gray paired node, then there are distinct residues $\xi,\eta\in\mathbb Z/2\mathbb Z$ such that all unpaired levels have residue $\xi$ and all gray-pair levels have residue $\eta$.
\end{lemma}
\begin{proof}
The two residue sets are nonempty and disjoint subsets of the two-element set $\mathbb Z/2\mathbb Z$.  Each must therefore be a singleton, and they must be the two different residue classes.
\end{proof}

The lemma deliberately excludes vacuous cases.  If there are no gray nodes, ordinary and modulo-2 separation are vacuous, whereas the uniform residue form still requires all unpaired levels to have one parity; the analogous issue occurs when there are no unpaired nodes.

\section{Prior results and main theorem}
\label{sec:prior}

The following prior results frame the theorem.

\begin{theorem}[Hale\v{s} et al.; separated coloring implies designability \citep{hales2017}]
A target whose tree representation admits a proper separated coloring is designable in the four-letter Watson--Crick base-pair model.
\end{theorem}

Hale\v{s} et al. also showed that saturated targets are designable over the four-letter alphabet precisely when their maximum paired degree is at most four, and that the motifs $m_5$ and $m_{3\bullet}$ are local obstructions.  Their separated-coloring theorem is sufficient but not necessary in general.

\begin{theorem}[Boury--Bulteau--Ponty]
Every $(m_5,m_{3\bullet})$-free target whose helices all have length at least three admits a modulo-2 separated coloring and can be designed in linear time \citep{boury2024,boury2025}.  More generally, their dynamic program decides modulo-$m$ separability in $O(n\,2^m)$ time.  Deciding ordinary separability remains NP-complete for motif-free targets without isolated base pairs \citep{boury2025}.
\end{theorem}

The same work gives a motif-free nonseparable construction with minimum helix length 2 that uses many isolated stacks.  A documented primary-source search through 19 August 2026 located no theorem or counterexample for the class with at most two isolated stacks; this negative search result is not by itself evidence of novelty.

\begin{theorem}[Main coloring theorem]
\label{thm:maincolor}
Every target $T\in\Ktwo$ admits a proper modulo-2 separated coloring.  In fact, the coloring has the uniform residue form: all unpaired nodes have one parity and all gray paired nodes have the other.
\end{theorem}

\begin{theorem}[Main design theorem]
\label{thm:maindesign}
Every target $T\in\Ktwo$ is designable in the four-letter Watson--Crick base-pair model with $\theta=0$.
\end{theorem}

The novel part of Theorem~\ref{thm:maincolor} is the bounded-resource extension of the uniform modulo-2 guarantee to targets containing one or two length-2 helices.  The all-long case is due to Boury et al.  For saturated members of the class, the designability conclusion already follows from Hale\v{s} et al.'s Result~R4; the present theorem supplies the stated modulo-2 certificate.  Theorem~\ref{thm:maindesign} then follows from the sequence construction and uniqueness argument in Sections~\ref{sec:sequence}--\ref{sec:notie}.

\begin{proposition}[Three isolated stacks can prevent modulo-2 separation]
\label{prop:three-short}
The motif-free target
\[
((.))((.))((.))
\]
has three isolated stacks and admits no proper modulo-2 separated coloring.  It nevertheless admits a proper ordinarily separated coloring and is uniquely designable.
\end{proposition}

This witness is closely related to Hale\v{s} et al.'s three-branch examples in their Figures~6--7 \citep{hales2017}.  In particular, their Figure~6 compares
\[
\texttt{(..)(..)(..)}\qquad\text{with}\qquad\texttt{((..))((..))((..))}.
\]
The latter is inflated and designable.  The literal above has the same three length-2 root branches but one, rather than two, unpaired positions in each hairpin.  Hale\v{s} et al. did not formulate modulo-2 nonseparability, although the proof below applies unchanged to their designable two-stutter because duplicate unpaired leaves inherit the same terminal level.  The exact literal and unique-design computation reported here use the one-unpaired-per-hairpin variant; the underlying three-branch geometry is adapted from that earlier example family.

\begin{proof}
Suppose a proper modulo-2 separated coloring existed.  The target has three unpaired nodes.  Its degree-3 root cannot have only black and white children, because properness permits at most one of each, so at least one root child is gray.  Lemma~\ref{lem:modtwo-uniform} therefore supplies residue classes $\xi$ for all unpaired nodes and $\eta$ for all gray pairs.

Each root-child helix has length 2 and terminates in a loop containing an unpaired node.  Lemma~\ref{lem:forced-residues} forces that terminal pair to lie at $\xi$ and to be non-gray.  The complete length-2 table in Lemma~\ref{lem:two-pair} is now a necessity statement: at a $\xi$ entry, only the words $\Bcol\Bcol$ and $\Wcol\Wcol$ can exit at $\xi$, so every root child begins non-gray; at an $\eta$ entry, only $\Gray\Bcol$ and $\Gray\Wcol$ can exit at $\xi$, so every root child begins gray.  In the first case the root would need three non-gray first colors, although properness permits only one black and one white.  In the second it would need three gray first colors, although properness permits at most two.  Both cases contradict properness.

For the caveat, the coloring with helix words $\Bcol\Bcol$, $\Wcol\Wcol$, and $\Gray\Bcol$ has gray level set $\{0\}$ and unpaired level set $\{-2,1,2\}$, so it is proper and ordinarily separated but not modulo-2 separated.  The color-to-sequence construction gives \texttt{GGACCCCAGGAGACU}.  The exact optimum-and-count dynamic program recorded with the release artifact finds target pair count 6, optimum 6, and optimum count 1.
\end{proof}

Thus the bound two is genuine for the modulo-2 coloring criterion used in the theorem, but not for ordinary separation or designability.  The proposition does not claim that every target with three isolated stacks fails modulo-2 separation.

\section{Loop types and required residues}
\label{sec:structural}

Call the innermost pair of a maximal helix its terminal pair.  The children of this pair describe the loop at the end of the helix.

\begin{lemma}[Loop types at a helix terminus]
\label{lem:endpoint-taxonomy}
The loop at the terminal pair of every nonroot helix belongs to exactly one of the following types:
\begin{itemize}
\item $E$: an empty terminal pair with no child (possible because $\theta=0$);
\item $L$: a hairpin, bulge, or internal-loop type with at least one unpaired position and zero or one outgoing helix;
\item $M$: a multiloop with no unpaired position and two or three outgoing helices.
\end{itemize}
The parent of the outer pair of a maximal helix is either the root, an $L$ node with one outgoing helix, or an $M$ node.
\end{lemma}

\begin{proof}
A terminal pair with one paired child and no unpaired child would continue the same helix by the helix-chain lemma, contradicting maximality.  If it has an unpaired child, the $m_{3\bullet}$ bound gives at most one paired child.  If it has no unpaired child, the $m_5$ bound permits zero, two, or three paired children.  The same maximality argument applies to the parent of a helix head.
\end{proof}

The loops of types $L$ and $M$ impose parity constraints in any coloring in the uniform residue form.

\begin{lemma}[Residues at terminal loops]
\label{lem:forced-residues}
Fix residues $\xi$ and $\eta=1-\xi$ in a proper coloring in the uniform residue form.
\begin{enumerate}
\item An $L$ loop lies at residue $\xi$; its own color and the colors of its paired children are non-gray.
\item An $M$ loop lies at residue $\eta$.
\item An $E$ loop has no forced residue.
\end{enumerate}
\end{lemma}

\begin{proof}
An unpaired child of an $L$ loop has the same level as the loop, so the loop lies at $\xi$.  A gray loop node or gray child would then place a gray paired node at $\xi$, contradicting the uniform residue condition.  For an $M$ loop, the exposed multiset has size at least three.  Properness permits at most one black and one white incidence, so some incidence is gray.  If the closing incidence is gray, the loop color is gray and its level is $\eta$; if a child incidence is gray, the child has the same inclusive residue as the loop, again forcing $\eta$.
\end{proof}

Thus an $L$ terminal requests exit residue $\xi$, an $M$ terminal requests $\eta$, and an $E$ terminal may be assigned either residue.  We choose $\xi$ for $E$.

\section{Feasible entry states for helix subtrees}
\label{sec:resource}

\subsection{Counting isolated stacks in a helix subtree}

Contract each maximal helix to one node in the tree of helix subtrees.  For a maximal helix $H$, let $\shortcount(H)$ be the number of isolated stacks (maximal helices of length 2) in the complete paired subtree rooted at $H$.  Let $K$ range over the outgoing child helices at the terminal loop of $H$.

\begin{lemma}[Subtree decomposition of the isolated-stack count]
\label{lem:short-decomp}
For every maximal helix $H$,
\[
\shortcount(H)
=
\mathbf 1_{\{|H|=2\}}
+
\sum_{K\text{ outgoing from }H}\shortcount(K).
\]
The child subtrees contain disjoint sets of maximal helices.
\end{lemma}

\begin{proof}
Every maximal helix in the paired subtree of $H$ is either $H$ itself or lies below exactly one paired child of the terminal pair.  The maximal-helix partition and the disjointness of sibling intervals make the sum disjoint.
\end{proof}

\begin{corollary}[Consequences of the global bound]
\label{cor:resource-bounds}
If the whole target lies in $\Ktwo$, then:
\begin{enumerate}
\item $\shortcount(H)\le2$ for every helix subtree;
\item at most two outgoing child subtrees of any loop contain an isolated stack;
\item if two outgoing child subtrees contain isolated stacks, then $H$ is not an isolated stack and hence $|H|\ge3$.
\end{enumerate}
The same bound holds among the root-child helix subtrees.
\end{corollary}

\begin{proof}
The first two claims follow from the global count bound and disjointness.  If two children already contain at least two isolated stacks and $H$ also had length 2, Lemma~\ref{lem:short-decomp} would give $\shortcount(H)\ge3$.  The target-class definition then gives $|H|\ge3$.
\end{proof}

\subsection[Feasible entry states F and Q]{Feasible entry states $F$ and $Q$}

Fix $\xi$ and let $\eta=1-\xi$.  An \emph{entry state} for a helix subtree is a pair consisting of the residue at which the helix is entered and the color of its first pair.  Define
\[
F=\{\xi\Bcol,\xi\Wcol,\eta\Bcol,\eta\Wcol,\eta\Gray\},
\qquad
Q=\{\xi\Bcol,\xi\Wcol,\eta\Gray\}.
\]
The state $\xi\Gray$ is impossible: a gray first pair does not change the entry residue and would therefore lie at $\xi$, whereas gray nodes must lie at $\eta$.

The proof uses the two sets as follows:
\begin{itemize}
\item a subtree containing no isolated stack is guaranteed to support every state in $F$;
\item a subtree containing one or two isolated stacks is guaranteed to support every state in $Q$.
\end{itemize}
At a $\xi$ entry, both sets allow black or white.  At an $\eta$ entry, the smaller guarantee $Q$ requires a gray first pair, whereas $F$ also allows black or white.

\begin{definition}[Feasible entry state]
For a helix subtree $H$, write $(r,c)\in\Accept(H)$ if the entire subtree can be colored with entry residue $r$ and first-pair color $c$ so that every exposure is proper, every gray pair has residue $\eta$, and every unpaired node has residue $\xi$.
\end{definition}

The proof establishes guaranteed subsets of $\Accept(H)$; it does not claim that $F$ and $Q$ list every feasible coloring.

\begin{theorem}[Inductive guarantee for helix subtrees]
\label{thm:resource-invariant}
For every helix subtree $H$ in a target from $\Ktwo$,
\[
\shortcount(H)=0 \implies F\subseteq\Accept(H),
\]
and
\[
1\le\shortcount(H)\le2 \implies Q\subseteq\Accept(H).
\]
\end{theorem}

The theorem is proved by well-founded induction in Sections~\ref{sec:transfers}--\ref{sec:globalcolor}.  Figure~\ref{fig:FQ} summarizes the two guarantees.

\begin{figure}[t]
\centering
\begin{tikzpicture}[font=\small,
state/.style={draw,rounded corners,minimum width=20mm,minimum height=8mm,align=center},
arr/.style={-{Latex[length=2.2mm]},thick}]
\node[state,fill=SoftBlue,draw=DeepBlue] (F) {$F$\\no isolated stack};
\node[state,fill=orange!10,draw=orange!70!black,right=35mm of F] (Q) {$Q$\\contains an isolated stack};
\draw[arr,DeepBlue] (Q) -- node[above]{subset} (F);
\node[below=9mm of F,align=center] (frows) {$\xi$: $\Bcol,\Wcol$\\$\eta$: $\Bcol,\Wcol,\Gray$};
\node[below=9mm of Q,align=center] (qrows) {$\xi$: $\Bcol,\Wcol$\\$\eta$: $\Gray$};
\draw[decorate,decoration={brace,amplitude=5pt,mirror},thick] ($(frows.south west)+(-2mm,-2mm)$) -- ($(qrows.south east)+(2mm,-2mm)$)
 node[midway,below=7mm,align=center,text width=9.0cm]{A child covered only by the $Q$ guarantee contains at least one isolated stack.  Since the whole target contains at most two, a loop or root has at most two such children.};
\end{tikzpicture}
\caption{The two sets of guaranteed entry states.  They differ only at an $\eta$ entry: $Q$ requires a gray first pair, whereas $F$ also allows black or white.}
\label{fig:FQ}
\end{figure}
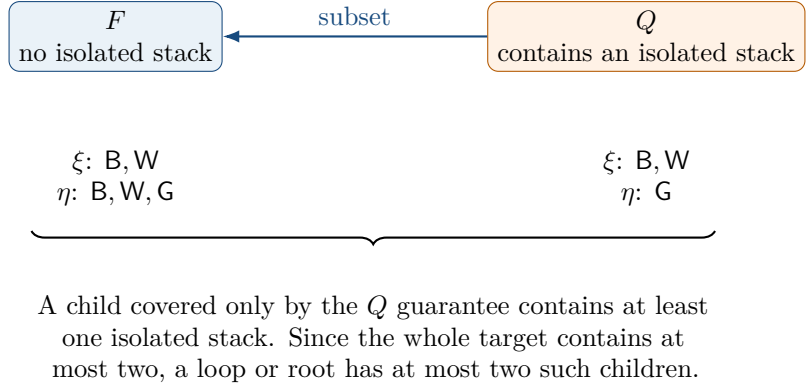

\section{Local helix transfers}
\label{sec:transfers}

A color word $c_1\cdots c_h$ on a helix is valid for entry residue $\varepsilon$ and requested exit $\tau$ if:
\begin{enumerate}
\item $c_1$ is the prescribed first color;
\item every adjacent exposure $\{\inv(c_i),c_{i+1}\}$ is proper;
\item every gray pair occurs at residue $\eta$; and
\item the inclusive residue of the final pair is $\tau$.
\end{enumerate}
When $\tau=\xi$, the construction additionally closes non-gray, as required by $L$ terminals.

\subsection{Color patterns for a length-2 helix}

\begin{lemma}[Length-2 transition table]
\label{lem:two-pair}
For a length-2 helix, the complete set of valid color words is:
\begin{center}
\begin{tabular}{cc}
\toprule
entry $\to$ exit & valid words \\
\midrule
$\xi\to\xi$ & $\Bcol\Bcol,\;\Wcol\Wcol$ \\
$\xi\to\eta$ & $\Bcol\Gray,\;\Wcol\Gray$ \\
$\eta\to\xi$ & $\Gray\Bcol,\;\Gray\Wcol$ \\
$\eta\to\eta$ & $\Bcol\Bcol,\;\Wcol\Wcol,\;\Gray\Gray$ \\
\bottomrule
\end{tabular}
\end{center}
\end{lemma}

\begin{proof}
The residues satisfy
\[
\level(x_1)\equiv\varepsilon+\delta(c_1),
\qquad
\level(x_2)\equiv\level(x_1)+\delta(c_2)\pmod 2.
\]
Require every gray occurrence to lie at $\eta$, the final residue to be $\tau$, and $\{\inv(c_1),c_2\}$ to satisfy the properness capacities.  Checking the nine color pairs gives exactly the table.
\end{proof}

In particular, an $\eta\to\xi$ isolated stack must start gray, and a $\xi\to\eta$ isolated stack must end gray.

\subsection{Long-helix transfer}

\begin{lemma}[Long transfer]
\label{lem:long-transfer}
Let $h\ge3$.  Suppose the first color is non-gray, or is gray at an $\eta$ entry.  For either requested exit residue $\tau\in\{\xi,\eta\}$, there is a valid color word of length $h$ preserving the first color; if $\tau=\xi$, the final color can be chosen non-gray.
\end{lemma}

\begin{proof}
Fix a non-gray color $a\in\{\Bcol,\Wcol\}$.  The local patterns
\[
(a,a),\quad(a,\Gray),\quad(\Gray,a),\quad(\Gray,\Gray)
\]
are proper, whereas $(a,\inv(a))$ is not.  If the first color is $a$, compare the constant word $a^h$ with a word having one gray insertion at the second or third position, chosen so that the gray lies at $\eta$.  The two words have opposite exit parities because a non-gray contributes $1$ modulo two and a gray contributes $0$.  The insertion position exists because $h\ge3$, and the $\xi$-exit choice cannot close gray.  If the first color is gray at an $\eta$ entry, use $\Gray^h$ for exit $\eta$ and $\Gray^{h-1}a$ for exit $\xi$.
\end{proof}

This sufficient transfer is the local engine in the all-long theorem of Boury et al. \citep{boury2024,boury2025}.  The at-most-two proof needs to select a particular non-gray-ending subfamily of their published transfer table when two short-helix demands meet.

\subsection[Non-gray-ending selection from Boury's transfer table]{Non-gray-ending selection from Boury's transfer table}

\begin{lemma}[Non-gray-ending long transfer]
\label{lem:strong-long}
Let $h\ge3$ and suppose the entry state lies in $Q$.  There exists a valid length-$h$ word that exits at $\eta$ and closes non-gray.
\end{lemma}

\begin{proof}
At a $\xi$ entry, let $a\in\{\Bcol,\Wcol\}$ be the prescribed non-gray first color.  At an $\eta$ entry, membership in $Q$ prescribes the first color $\Gray$; here $a$ denotes the free trailing non-gray color, which we fix as $\Bcol$ to make the construction deterministic.  Use:
\begin{center}
\begin{tabular}{lll}
\toprule
entry and first color & parity of $h$ & word \\
\midrule
$\xi$, first $a$ & odd & $a^h$ \\
$\xi$, first $a$ & even & $a\Gray a^{h-2}$ \\
$\eta$, first $\Gray$ & odd & $\Gray a^{h-1}$ \\
$\eta$, first $\Gray$ & even & $\Gray\Gray a^{h-2}$ \\
\bottomrule
\end{tabular}
\end{center}
Each word preserves the first color and closes with $a$.  Counting non-gray contributions modulo two gives exit $\eta$.  The displayed gray positions remain at $\eta$, and every adjacency is one of the proper patterns above.  These are precisely the relevant columns of Boury et al.'s WABI Figure~8 and journal Figure~12, with black--white symmetry made explicit \citep{boury2024,boury2025}.  In particular, the odd $\eta$-entry words at length three are $\Gray\Bcol\Bcol$ and $\Gray\Wcol\Wcol$.
\end{proof}

The words $\Gray\Bcol\Bcol$ and $\Gray\Wcol\Wcol$ are therefore prior local ingredients, not new patterns.  The contribution below is the global counting argument that identifies when the incoming helix must be long and hence when this published non-gray-ending option is guaranteed to be available.

\section{Color assignments at loops and the root}
\label{sec:allocations}

\subsection{Assignments at terminal loops}

Let $r$ be the number of outgoing child helix subtrees that contain at least one isolated stack.  By Corollary~\ref{cor:resource-bounds}, $r\le2$.

For a terminal pair of type $E$, there is no outgoing child and $X(v)=\{\inv(a)\}$ is proper.  For a terminal loop of type $L$, the current helix exits at $\xi$ and ends in a non-gray color $a$.  If there is one outgoing helix, assign its first pair the same color $a$, giving exposure $\{\inv(a),a\}$.  The child enters at $\xi$, where both $F$ and $Q$ allow $a$.

At a multiloop of type $M$, the number of outgoing helices is $d\in\{2,3\}$ and the required exit residue is $\eta$.

\paragraph{No child subtree contains an isolated stack ($r=0$).}
Use the following first-pair colors:
\begin{center}
\begin{tabular}{ccc}
\toprule
closing color & $d=2$ child colors & $d=3$ child colors \\
\midrule
$\Bcol$ & $\Bcol,\Gray$ & $\Bcol,\Gray,\Gray$ \\
$\Wcol$ & $\Wcol,\Gray$ & $\Wcol,\Gray,\Gray$ \\
$\Gray$ & $\Bcol,\Wcol$ & $\Bcol,\Wcol,\Gray$ \\
\bottomrule
\end{tabular}
\end{center}
Every child subtree is covered by $F$ and therefore supports each displayed state at an $\eta$ entry.

\paragraph{One child subtree contains an isolated stack ($r=1$).}
Assign gray to that child and use the following colors for the other children:
\begin{center}
\begin{tabular}{ccc}
\toprule
closing color & other color for $d=2$ & other colors for $d=3$ \\
\midrule
$\Bcol$ & $\Bcol$ & $\Bcol,\Gray$ \\
$\Wcol$ & $\Wcol$ & $\Wcol,\Gray$ \\
$\Gray$ & $\Bcol$ & $\Bcol,\Wcol$ \\
\bottomrule
\end{tabular}
\end{center}
The child containing an isolated stack supports $\eta\Gray$ through $Q$; every other child is covered by $F$.

\paragraph{Two child subtrees contain isolated stacks ($r=2$).}
The incoming helix is long by Corollary~\ref{cor:resource-bounds}.  Apply Lemma~\ref{lem:strong-long} so that it reaches residue $\eta$ but ends in a non-gray color $a$.  Assign gray to both child helices.  If $d=3$, assign the remaining child, which contains no isolated stack, the color $a$.  The exposures are permutations of
\[
\{\inv(a),\Gray,\Gray\}
\quad\text{or}\quad
\{\inv(a),\Gray,\Gray,a\}.
\]
For $a=\Bcol$ these are $\{\Wcol,\Gray,\Gray\}$ and $\{\Wcol,\Gray,\Gray,\Bcol\}$; the white case is symmetric.  Each exposure has at most one black, one white, and two gray incidences.

\begin{figure}[t]
\centering
\begin{tikzpicture}[font=\small,
loop/.style={circle,draw=DeepBlue,fill=SoftBlue,minimum size=9mm},
helix/.style={draw,rounded corners,minimum width=18mm,minimum height=8mm,align=center},
lab/.style={font=\footnotesize},
arr/.style={thick}]
\node[helix,fill=gray!8] (in) {incoming long helix\\$\Gray\Bcol\Bcol$};
\node[loop,below=8mm of in] (m) {multiloop};
\node[helix,fill=orange!9,below left=13mm and 15mm of m] (q1) {child subtree\\with isolated stack};
\node[helix,fill=orange!9,below right=13mm and 15mm of m] (q2) {child subtree\\with isolated stack};
\draw[arr] (in)--node[right,lab]{ends $\Bcol$} (m);
\draw[arr] (m)--node[left,lab]{first color $\Gray$} (q1);
\draw[arr] (m)--node[right,lab]{first color $\Gray$} (q2);
\node[right=20mm of m,align=left] (ex) {$X(M)=\{\Wcol,\Gray,\Gray\}$\\one white, two gray\\\textcolor{GoodGreen}{proper}};
\draw[-{Latex[length=2mm]},DeepBlue] (m)--(ex);
\end{tikzpicture}
\caption{The case with two child subtrees containing isolated stacks.  The long incoming helix uses $\Gray\Bcol\Bcol$, so it ends black rather than gray.  The two child helices can then both start gray without violating properness.}
\label{fig:two-child}
\end{figure}
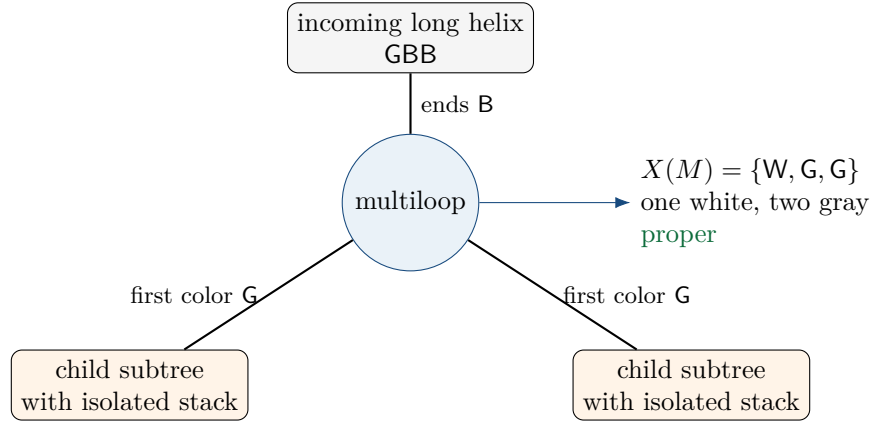

Figure~\ref{fig:two-child} depicts the load-bearing $r=2$ allocation.  The table rows are assigned to the actual child helices in backbone order.  When $r=1$ or $r=2$, the isolated-stack count decomposition identifies exactly which children must receive gray; the remaining colors are assigned deterministically to the other children.

\subsection{Assignments at the root}

If the root has no paired child, the target is all-unpaired.  Choose $\xi=0$ and use the empty coloring.  The root exposure is empty, every unpaired node has level zero, and there are no gray pairs.

Assume the root has at least one paired child.  If it has an unpaired child, $m_{3\bullet}$-freeness gives paired degree at most two and its level forces $\xi=0$.  Assign the child first-pair-color row $\Bcol$ at degree one or $\Bcol,\Wcol$ at degree two.  If the root has no unpaired child and paired degree at most two, either naming of the two residues is mathematically valid; to fix all free choices deterministically, choose $\xi=0$, $\eta=1$, and use the same child first-pair-color rows.

If the root has no unpaired child and degree $d\in\{3,4\}$, set $\eta=0$ and $\xi=1$.  Let $r$ be the number of root-child subtrees that contain an isolated stack; $r\le2$.  Assign gray to each such child and fill the remaining positions as follows:
\begin{center}
\begin{tabular}{cccc}
\toprule
root degree & $r=0$ & $r=1$ & $r=2$ \\
\midrule
$3$ & $\Bcol,\Wcol,\Gray$ & $\Gray,\Bcol,\Wcol$ & $\Gray,\Gray,\Bcol$ \\
$4$ & $\Bcol,\Wcol,\Gray,\Gray$ & $\Gray,\Bcol,\Wcol,\Gray$ & $\Gray,\Gray,\Bcol,\Wcol$ \\
\bottomrule
\end{tabular}
\end{center}
Each child subtree containing an isolated stack supports $\eta\Gray$ through $Q$, and each remaining child is covered by $F$.

\begin{lemma}[Complete local color assignment]
\label{lem:complete-alloc}
At every terminal loop and at the root, the assignments above give each actual outgoing helix a feasible entry state and keep the actual exposed multiset proper.
\end{lemma}
\begin{proof}
Lemma~\ref{lem:endpoint-taxonomy} covers all terminal-loop types.  The $E$ and $L$ cases are direct.  At a multiloop, $r$ is $0$, $1$, or $2$.  The tables cover every possible closing color in the first two cases; Lemma~\ref{lem:strong-long} supplies a non-gray closing color in the third.  The color counts are explicit in every row.  The root cases follow from the motif degree bounds and the bound $r\le2$.
\end{proof}

\section{Global coloring by induction}
\label{sec:globalcolor}

The local colorings must be combined without allowing choices in one child subtree to affect another.  We use the following induction invariant, stated so that it remains valid after the rest of the target is colored.

\begin{definition}[Subtree coloring invariant]
For a maximal helix $H$, an entry residue $\varepsilon$, and a first color $c_1$, the induction hypothesis gives a coloring defined exactly on the paired nodes in the helix subtree of $H$.  For every total coloring that agrees with it on that subtree and has the specified entry residue, the following statements hold:
\begin{enumerate}
\item the first pair of $H$ has color $c_1$, and the selected color word is used on the actual pairs of $H$;
\item every paired exposure in the subtree is proper;
\item every gray paired node in the subtree has residue $\eta$;
\item every unpaired position in the subtree has residue $\xi$;
\item every outgoing child helix starts with the color assigned at its parent loop; and
\item the terminal residue and any required non-gray final color are satisfied.
\end{enumerate}
\end{definition}

The domains of different child subtrees are disjoint, and the pairs in the current helix are disjoint from all child subtrees.  Their union is the complete parent subtree.  A path from the root to a node in one child does not enter a sibling subtree, so colors in a sibling do not change that node's level.  The child subtrees interact only through the exposed multiset at their common parent, which is fixed before the recursive calls.

\begin{proof}[Proof of Theorem~\ref{thm:resource-invariant}]
Proceed by well-founded induction on the number of paired nodes in the helix subtree of $H$.  Fix $\xi$, set $\eta=1-\xi$, and assume the requested entry state lies in $F$ when $\shortcount(H)=0$, or in $Q$ when $\shortcount(H)>0$.

Let $v$ be the terminal pair of $H$.  By Lemma~\ref{lem:endpoint-taxonomy}, the loop at $v$ is of type $E$, $L$, or $M$.  Request exit residue $\xi$ for $E$ and $L$, and $\eta$ for $M$.

If $|H|=2$, use the exact table in Lemma~\ref{lem:two-pair}.  Such a helix belongs to a subtree containing an isolated stack, so the entry state is covered by $Q$.  If $|H|\ge3$, use Lemma~\ref{lem:long-transfer}, except when $v$ is a multiloop with two child subtrees containing isolated stacks.  In that case use Lemma~\ref{lem:strong-long} to reach $\eta$ with a non-gray final pair.

Apply Lemma~\ref{lem:complete-alloc} to assign the first-pair colors of the child helices.  Each child subtree contains fewer paired nodes than the parent.  Its isolated-stack count determines whether the required entry state comes from $F$ or $Q$, and the local assignment lemma proves that the selected color is allowed.  Apply the induction hypothesis to each child.

Finally combine the color word on $H$ with the colorings of the disjoint child subtrees.  Internal helix exposures are proper by the selected word, the terminal exposure is proper by the local assignment, and all child exposures are proper by induction.  The helix transfer and child hypotheses put every gray pair at $\eta$.  Unpaired positions occur only at loops of type $L$ and therefore have residue $\xi$.  The combined coloring satisfies the stated invariant under every extension outside the subtree.  This proves the $F$ guarantee when $\shortcount(H)=0$ and the $Q$ guarantee when $1\le\shortcount(H)\le2$.
\end{proof}

\begin{proof}[Proof of Theorem~\ref{thm:maincolor}]
If the root has degree zero, use the empty coloring.  Otherwise choose one of the root assignments above.  For each root-child helix, determine whether its subtree contains an isolated stack, take the corresponding state from $F$ or $Q$, and apply Theorem~\ref{thm:resource-invariant}.  The root-child paired subtrees are disjoint and cover all target pairs, so their colorings combine into a total coloring.

The root exposure is proper by construction.  Every nonroot exposure is proper by the induction invariant of its unique subtree.  Every gray pair has residue $\eta$.  Every nonroot unpaired position lies below a loop of type $L$ and has residue $\xi$; a root-unpaired position has level zero, and in that case the root construction chooses $\xi=0$.  Therefore the coloring is proper and has the uniform residue form of modulo-2 separation.
\end{proof}

\subsection{Algorithm and complexity}

The proof gives a deterministic recursive algorithm after the tree representation, maximal helices, child order, isolated-stack counts, root residue convention, symmetric non-gray choices, and gray-pair orientations have all been fixed as above.  One version is shown in Algorithm~\ref{alg:color}.

\begin{algorithm}[t]
\caption{Coloring a helix subtree}
\label{alg:color}
\DontPrintSemicolon
\KwIn{helix $H$; residues $\xi,\eta$; entry residue; first color; count $\shortcount(H)\le2$}
\KwOut{a coloring satisfying the subtree induction invariant}
classify the terminal loop as $E$, $L$, or $M$\;
request exit $\xi$ for $E/L$, and $\eta$ for $M$\;
identify the child subtrees that contain isolated stacks\;
\eIf{$H$ is long, the terminal is $M$, and two child subtrees contain isolated stacks}{
  choose the additional $\eta$-exit transfer with a non-gray final pair\;
}{
  choose the ordinary length-2 or long-helix transfer\;
}
assign first-pair colors to the actual child helices\;
\ForEach{outgoing child helix $K$}{
  select $F$ or $Q$ according to $\shortcount(K)$\;
  recursively color $K$ using its assigned first color\;
}
combine the color word on $H$ with the disjoint child colorings\;
return the coloring and the induction invariant\;
\end{algorithm}

With adjacency lists and helix membership stored explicitly, each paired node, unpaired node, helix, and child edge is processed $O(1)$ times.  Isolated-stack counts are computed bottom-up and colors are assigned top-down.  Thus the construction can be implemented in $O(n)$ time and $O(n)$ space.  The Lean development proves existence and correctness and uses classical choice for several finite selections; it is not presented as extracted executable code.

\section{From a coloring to an RNA sequence}
\label{sec:sequence}

Let $T$ have a proper coloring.  Construct a sequence $w$ top-down:
\begin{enumerate}
\item every target-unpaired position receives $\A$;
\item a black pair $[i,j]$ receives $(w_i,w_j)=(\G,\C)$;
\item a white pair receives $(w_i,w_j)=(\C,\G)$;
\item a gray pair receives either $(\A,\U)$ or $(\U,\A)$, with the following orientation rule:
  \begin{itemize}
  \item at the root or below a non-gray parent, gray siblings receive distinct left letters from $\{\A,\U\}$ in backbone order;
  \item below a gray parent, its at most one gray child copies the parent's left letter;
  \end{itemize}
  and every right endpoint is the complement of its left endpoint.
\end{enumerate}
Properness makes the rule well-defined: the root or a non-gray parent has at most two gray children, while a gray parent has at most one gray child.

\begin{lemma}[Local nucleotide distinctness]
\label{lem:local-distinct}
At the root, the left letters of paired children are pairwise distinct.  At every nonroot paired node, the left letters of paired children are pairwise distinct and differ from the parent's right letter.
\end{lemma}

\begin{proof}
For a black parent, properness permits at most one black child and no white child; possible child-left letters are therefore $\G,\A,\U$, while the parent right letter is $\C$.  The white case is symmetric.  For a gray parent with left letter $L\in\{\A,\U\}$, properness permits at most one black, one white, and one gray child.  Their possible left letters are $\G,\C,L$, while the parent right letter is $\comp(L)$.  At the root, the possible child-left letters are $\G,\C,\A,\U$, with distinct orientations assigned to gray siblings.
\end{proof}

The copying rule below a gray parent is essential.  For the nested all-gray target $(())$, copying yields $\A\A\U\U$, whereas the opposite orientation gives $\A\U\A\U$, which is also compatible with the disjoint structure $()()$ with the same pair count.

\section{Nucleotide-count bound}
\label{sec:inventory}

Let $g$ be the number of gray target pairs, $q$ the number of non-gray target pairs, and $u$ the number of target-unpaired positions.

\begin{lemma}[Nucleotide-count bound]
\label{lem:inventory}
The assigned sequence satisfies
\[
\#\U=g,
\qquad
\#\G=\#\C=q,
\qquad
\#\A=g+u.
\]
Every compatible noncrossing structure $S$ on the same sequence satisfies
\[
|S|\le g+q=|T|.
\]
If equality holds, every $\U$, $\C$, and $\G$ position is paired in $S$.
\end{lemma}

\begin{proof}
Each gray target pair contributes one $\A$ and one $\U$; each black or white target pair contributes one $\G$ and one $\C$; each target-unpaired position contributes one $\A$.  Every A--U pair in $S$ consumes a distinct $\U$, and every C--G pair consumes a distinct $\C$.  Hence
\[
|S|\le \#\U+\#\C=g+q=|T|.
\]
If equality holds, all $\U$ and all $\C$ positions are used.  Every used $\C$ has a distinct $\G$ partner, and $\#\G=\#\C$, so every $\G$ is used as well.
\end{proof}

Lemma~\ref{lem:inventory} proves that the target has the maximum possible number of base pairs.  It does not yet exclude a distinct fold with the same number of pairs.

\section{Uniqueness on the paired restriction}
\label{sec:saturated}

This section proves a general uniqueness theorem for the paired restriction of a target, using the local-distinctness condition of Lemma~\ref{lem:local-distinct}.  It is independent of the short-helix count.

\subsection{Saturable sequences and atomic saturable designs}

A finite nucleotide word is \emph{saturable} if it admits a compatible noncrossing perfect matching.  The empty word is saturable.  A nonempty saturable sequence is \emph{atomic} if no nonempty proper prefix is saturable.  Following Hale\v{s} et al., an \emph{atomic saturable design} is an atomic saturable sequence together with a saturated target for which it is the unique compatible noncrossing perfect matching.

\begin{lemma}[Adjacent cancellation and suffix cancellation]
\label{lem:suffix-cancel}
A word is saturable if and only if it can be reduced to the empty word by repeatedly deleting adjacent complementary letters.  Consequently, if $xy$ and $x$ are saturable, then $y$ is saturable.
\end{lemma}

\begin{proof}
Suppose a word has a compatible noncrossing perfect matching.  Choose a matched pair of minimum span.  If its endpoints were not adjacent, the position immediately inside the left endpoint would have a mate strictly inside the chosen pair, producing a smaller span.  Therefore the chosen endpoints are adjacent and complementary.  Delete them, compress the remaining positions, and induct.

Conversely, reverse a sequence of adjacent complementary deletions.  Reinsert each deleted adjacent pair and match its two positions.  The new pair is disjoint from or nested within every previously restored pair, so the resulting perfect matching is compatible and noncrossing.

For the suffix statement, define the reduced stack normal form $R(z)$ by scanning left to right and cancelling the next letter whenever it complements the stack top.  Adjacent complementary deletion preserves $R$, and a word reduces to the empty word exactly when $R(z)$ is empty.  The stack identity
\[
R(xy)=R(R(x)y)
\]
then gives $R(y)=R(xy)=\varnothing$ whenever $R(x)=\varnothing$.  Thus $y$ is saturable.
\end{proof}

Equivalently, one may encode $\A,\U$ as inverse generators and $\C,\G$ as another inverse pair in a free group; suffix cancellation is then ordinary group cancellation.  The stack formulation above is sufficient for the proof.

\begin{lemma}[Atomic endpoint characterization]
\label{lem:atomic-endpoint}
A nonempty saturable word $z$ is atomic if and only if every compatible noncrossing perfect matching of $z$ contains the outer pair joining the first and last positions.
\end{lemma}

\begin{proof}
If $z$ is atomic and the first position is paired to a position before the end, noncrossingness and saturation make the enclosed prefix a nonempty proper saturable prefix.  Conversely, if a nonempty proper prefix $p$ is saturable, write $z=ps$.  Lemma~\ref{lem:suffix-cancel} makes $s$ saturable, and concatenating perfect matchings of $p$ and $s$ yields a perfect matching of $z$ omitting the outer pair.
\end{proof}

\subsection{Concatenation and wrapping}

\begin{lemma}[Concatenating atomic saturable designs]
\label{lem:concat-atomic}
Let $z_1,\dots,z_t$ be atomic saturable designs, with $t\ge1$, and suppose their first letters are pairwise distinct.  Then the concatenation $z_1\cdots z_t$ has the concatenation of the target matchings as its unique compatible noncrossing perfect matching.
\end{lemma}

\begin{proof}
Let $P$ be a compatible noncrossing perfect matching of the concatenation.  If $P$ contains the outer pair of every block, those pairs isolate the blocks and uniqueness holds blockwise.

Otherwise choose the leftmost block $z_i$ whose first position is not paired to its own last position.  Earlier blocks are isolated.  Let the first position of $z_i$ pair to a position in block $z_j$, $j\ge i$.  If it pairs to the last position of $z_j$, then $j>i$ and unique complementation forces the first letters of $z_i$ and $z_j$ to be equal, contradicting pairwise distinctness.  If it pairs inside $z_j$, the enclosed word
\[
z_i z_{i+1}\cdots z_{j-1} z_j[1,k]
\]
is saturable.  Repeated suffix cancellation removes the complete saturable blocks $z_i,\dots,z_{j-1}$, leaving a nonempty proper saturable prefix of the atomic word $z_j$, again a contradiction.
\end{proof}

\begin{lemma}[Wrapping atomic saturable designs]
\label{lem:wrap-atomic}
Let $z_1,\dots,z_t$ be atomic saturable designs, possibly with $t=0$.  Let $a,b$ be complementary letters.  Suppose the first letters of the $z_i$ are pairwise distinct and none equals $b$.  Then
\[
W=a z_1\cdots z_t b
\]
is an atomic saturable design for the target consisting of the outer pair $(a,b)$ around the concatenated child targets.
\end{lemma}

\begin{proof}
For $t=0$, the two-letter complementary word is immediate.  Assume $t\ge1$.  The displayed target proves saturability.  Suppose $W$ has a nonempty proper saturable prefix and choose one of minimum length.  Its length cannot be one plus the total length of complete initial child blocks, because those lengths are odd while a perfectly matched word has even length.  Hence the prefix ends strictly inside some $z_j$.

In a perfect matching of this shortest prefix, the initial letter $a$ must pair to the final position; otherwise the prefix ending at its partner would be shorter and saturable.  The final letter is therefore $b$.  Removing this outer pair leaves
\[
z_1\cdots z_{j-1}z_j[1,k-1]
\]
saturable.  Repeated suffix cancellation removes the complete blocks, so $z_j[1,k-1]$ is saturable.  Atomicity forces $k-1=0$, making the first letter of $z_j$ equal to $b$, a contradiction.  Therefore $W$ is atomic.  Lemma~\ref{lem:atomic-endpoint} forces the outer pair in every perfect matching, and Lemma~\ref{lem:concat-atomic} gives uniqueness inside it.
\end{proof}

\begin{theorem}[Saturated local-distinctness theorem]
\label{thm:saturated-unique}
Let $R$ be a saturated target and $z$ a compatible sequence.  Suppose that:
\begin{enumerate}
\item at the root, the left letters of paired children are pairwise distinct; and
\item at every nonroot pair $[i,j]$, the left letters of paired children are pairwise distinct and none equals $z_j$.
\end{enumerate}
Then $R$ is the unique compatible noncrossing perfect matching of $z$.
\end{theorem}

\begin{proof}
Induct upward through the interval tree.  A leaf pair is Lemma~\ref{lem:wrap-atomic} with $t=0$.  At a nonroot pair $[i,j]$, saturation makes the interior a concatenation of paired-child subwords.  By induction these are atomic saturable designs.  Compatibility makes $z_i$ and $z_j$ complementary, and the local hypotheses are precisely those of Lemma~\ref{lem:wrap-atomic}; hence the whole interval subword is an atomic saturable design.  At the root, the complete word is a concatenation of root-child atomic saturable designs with pairwise distinct first letters, so Lemma~\ref{lem:concat-atomic} gives uniqueness.
\end{proof}

Lemma~\ref{lem:local-distinct} shows that the sequence produced from every proper coloring satisfies the hypotheses of Theorem~\ref{thm:saturated-unique} on its paired restriction.

\section{Excluding alternative optimal folds}
\label{sec:notie}

It remains to show that a different compatible structure with the same number of pairs cannot use a position that is unpaired in the target.

\subsection{G--C prefix balance and levels}

For the constructed sequence $w$, define the inclusive prefix balance
\[
\Lambda(k)=\#\G\text{ in }w[1,k]-\#\C\text{ in }w[1,k],
\qquad \Lambda(0)=0.
\]

\begin{lemma}[Levels equal prefix balances]
\label{lem:level-balance}
For every target pair $v=[i,j]$,
\[
\Lambda(i)=\level(v).
\]
For every target-unpaired position $k$, $\Lambda(k)$ equals the level of its unpaired interval-tree node.  If $v=[i,j]$ is gray, then
\[
\Lambda(j)=\Lambda(i)=\level(v).
\]
\end{lemma}

\begin{proof}
Every complete target subtree contains equal numbers of $\G$ and $\C$: a black or white pair contributes one of each, while gray pairs and unpaired $\A$ positions contribute neither.  Completed earlier siblings therefore contribute zero to the running balance.  At a paired left endpoint, the only uncancelled nonzero contributions are the left endpoints on the open ancestor path.  Black contributes $+1$, white $-1$, and gray $0$, exactly reproducing the inclusive level.  An unpaired $\A$ changes nothing and inherits the enclosing level.  A gray pair and its complete interior contribute zero, so its right and left endpoints have the same balance.
\end{proof}

\begin{lemma}[Level-imbalance obstruction]
\label{lem:imbalance}
Let $S$ be a compatible noncrossing structure of $w$.  If $S$ contains an A--U pair between positions $a<b$ with $\Lambda(a)\ne\Lambda(b)$, then $S$ leaves at least one $\G$ or $\C$ position unpaired.
\end{lemma}

\begin{proof}
No position strictly between $a$ and $b$ can pair outside that interval without crossing $(a,b)$.  Since the endpoints have zero G-minus-C weight,
\[
\#\G\text{ in }(a,b)-\#\C\text{ in }(a,b)
=
\Lambda(b)-\Lambda(a)\ne0.
\]
If every interior $\G$ and $\C$ were paired, compatibility and noncrossing closure inside the arc would pair them bijectively, forcing equal counts.  Thus some $\G$ or $\C$ remains unpaired.
\end{proof}

\subsection{Uniqueness among optimal folds}

\begin{theorem}[No distinct optimal structure]
\label{thm:no-tie}
Let $T$ have a proper separated coloring and let $w$ be the sequence assigned in Section~\ref{sec:sequence}.  If a compatible noncrossing structure $S$ satisfies $|S|=|T|$, then $S=T$.
\end{theorem}

\begin{proof}
By Lemma~\ref{lem:inventory}, equality in pair count requires every $\U$, $\C$, and $\G$ position to be paired.

Suppose a position $i$ unpaired in $T$ is paired to $j$ in $S$.  Then $w_i=\A$ and compatibility gives $w_j=\U$.  Every $\U$ belongs to a gray target pair $v$.  Lemma~\ref{lem:level-balance} identifies $\Lambda(i)$ with the target-unpaired level and $\Lambda(j)$ with the gray-pair level, regardless of which endpoint of $v$ carries $\U$.  Separation makes these levels different.  After ordering the endpoints as $a=\min\{i,j\}$ and $b=\max\{i,j\}$, Lemma~\ref{lem:imbalance} says that $S$ leaves a $\G$ or $\C$ unpaired, contradicting the equality case of Lemma~\ref{lem:inventory}.  Hence every target-unpaired position remains unpaired in $S$.

The two structures have equal pair count and therefore equal numbers of unpaired positions.  Their unpaired-position sets are equal.  Delete that common set and compress the remaining positions in backbone order.  The restrictions of $T$ and $S$ are saturated, compatible, and noncrossing on the same restricted word.  Deleting target-unpaired singleton nodes does not change paired parent-child relations, so Lemma~\ref{lem:local-distinct} transfers to the restricted target.  Theorem~\ref{thm:saturated-unique} makes that target the unique compatible perfect matching.  Thus the restrictions agree, and restoring the common unpaired positions gives $S=T$.
\end{proof}

\begin{proof}[Proof of Theorem~\ref{thm:maindesign}]
Theorem~\ref{thm:maincolor} gives a proper coloring in the uniform residue form, hence a proper separated coloring.  Construct the sequence in Section~\ref{sec:sequence}.  Lemma~\ref{lem:inventory} shows that no compatible structure has more pairs than the target, and Theorem~\ref{thm:no-tie} shows that no distinct compatible structure can attain the same pair count.  Therefore the target is the unique compatible noncrossing structure with the maximum number of base pairs and hence the unique minimum-energy fold.
\end{proof}

\section{Worked examples}
\label{sec:examples}

The universal proof does not depend on finite examples, but concrete instances clarify the two new coloring cases.

\subsection{Two isolated stacks at a degree-three root}

Consider
\[
T_1=\texttt{(())(())((()))},
\]
which has two isolated stacks at the root and one root helix of length three.\footnote{Because $T_1$ is saturated and has maximum paired degree at most four, its designability already follows from Hale\v{s} et al.'s Result~R4 \citep{hales2017}.  It is included here only to illustrate the bounded-two coloring construction and its formalized example checker.}  The root has no unpaired position and degree three, so choose $\eta=0$, $\xi=1$, and the root row for two such children
\[
\Gray,\Gray,\Bcol.
\]
Each isolated stack enters at $\eta$ through gray and terminates at $E$, so it uses the $\eta\to\xi$ word $\Gray\Bcol$.  The long helix enters through black and uses a length-three transfer to its $E$ terminal.  One deterministic orientation gives
\[
w_1=\texttt{\TOneSequenceLiteral}.
\]
The target has seven pairs.  Exact Nussinov dynamic programming \citep{nussinov1980} gives optimum score seven and optimum count one.

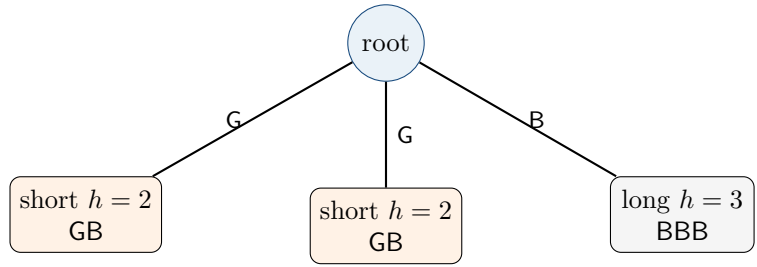
\begin{figure}[t]
\centering
\begin{tikzpicture}[font=\small,
helix/.style={draw,rounded corners,minimum width=19mm,minimum height=10mm,align=center},
root/.style={circle,draw=DeepBlue,fill=SoftBlue,minimum size=10mm},
arr/.style={thick}]
\node[root] (r) {root};
\node[helix,fill=orange!10,below left=14mm and 26mm of r] (h1) {short $h=2$\\$\Gray\Bcol$};
\node[helix,fill=orange!10,below=14mm of r] (h2) {short $h=2$\\$\Gray\Bcol$};
\node[helix,fill=gray!8,below right=14mm and 26mm of r] (h3) {long $h=3$\\$\Bcol\Bcol\Bcol$};
\draw[arr] (r)--node[left,font=\footnotesize]{$\Gray$} (h1);
\draw[arr] (r)--node[right,font=\footnotesize]{$\Gray$} (h2);
\draw[arr] (r)--node[right,font=\footnotesize]{$\Bcol$} (h3);
\node[below=11mm of h2,align=center,text width=10cm] {
Target: \texttt{(())(())((()))}\qquad
Sequence: \texttt{\TOneSequenceLiteral}\\
root exposure $\{\Gray,\Gray,\Bcol\}$ is proper; optimum pair count $7$, unique optimum.};
\end{tikzpicture}
\caption{A degree-three root with two child subtrees containing isolated stacks.  The root row assigns gray to the two child subtrees containing isolated stacks and black to the remaining long-helix subtree.}
\label{fig:root-example}
\end{figure}

Figure~\ref{fig:root-example} makes the degree-3 root assignment explicit: the two gray first-pair colors are assigned to the two isolated-stack subtrees.

\subsection{Two isolated-stack subtrees at an internal multiloop}

The target
\[
T_2=\texttt{(((((((())(()))))((())))))}
\]
has five maximal helices with lengths $(3,3,2,2,3)$.  Its outer length-3 helix ends at an $M$ loop with two long children.  One of those long children ends at a second $M$ loop whose two children are the two isolated stacks.  At that internal loop, both child subtrees contain an isolated stack, so the incoming length-3 helix uses the strengthened transfer
\[
\Gray\Bcol\Bcol
\]
and closes black.  The two short children each receive gray and use $\Gray\Bcol$.  The root has paired degree one and no unpaired child, so the deterministic small-root convention fixes
\[
\xi=0,\qquad \eta=1.
\]
The helix words, in the order outer, critical, two short children, and side child, are
\[
\Bcol\Bcol\Bcol,\quad
\Gray\Bcol\Bcol,\quad
\Gray\Bcol,\quad
\Gray\Bcol,\quad
\Bcol\Bcol\Bcol,
\]
with integer levels $(1,2,3)$, $(3,4,5)$, $(5,6)$, $(5,6)$, and $(4,5,6)$, respectively.  The deterministic orientation rule then gives
\[
w_2=\texttt{\TTwoSequenceLiteral}.
\]
The target has thirteen pairs; exact Nussinov dynamic programming gives optimum score thirteen and optimum count one.

The final review also reported \texttt{GGGGGAGGCCCCGGUCCAGGCCUCCC}.  That sequence is compatible and uniquely optimal for the unchanged target, but its helix words are $\Bcol\Bcol\Bcol$, $\Bcol\Bcol\Gray$, $\Bcol\Bcol$, $\Wcol\Wcol$, and $\Gray\Bcol\Bcol$.  It therefore does not instantiate the stated $F/Q$ recursion or the internal strengthened-transfer case.  We use the independently reconstructed sequence above, which does.

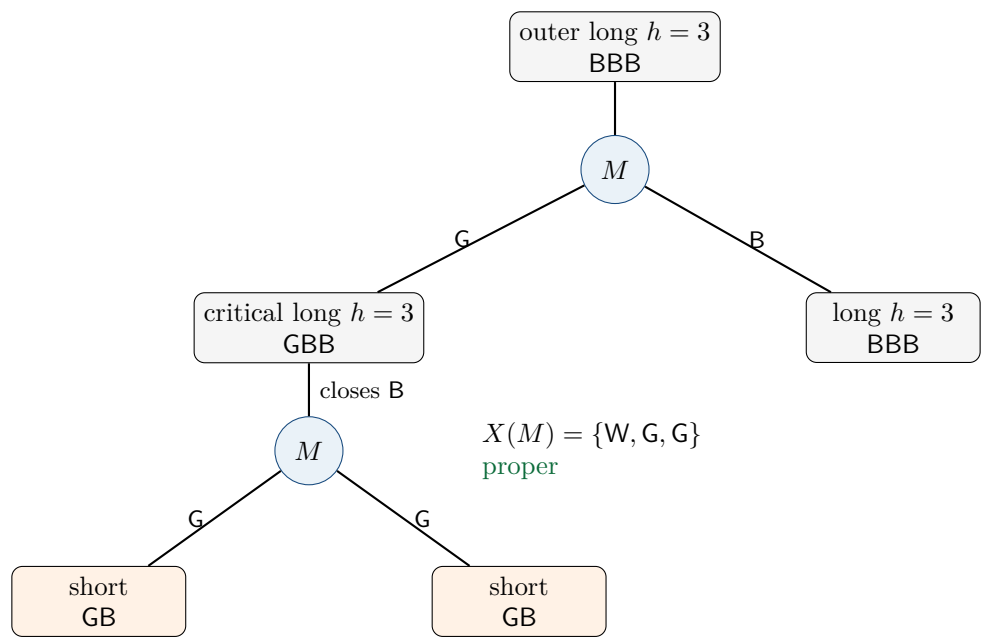
\begin{figure}[H]
\centering
\begin{tikzpicture}[font=\small,
helix/.style={draw,rounded corners,minimum width=23mm,minimum height=9mm,align=center},
loop/.style={circle,draw=DeepBlue,fill=SoftBlue,minimum size=9mm},
arr/.style={thick}]
\node[helix,fill=gray!8] (outer) {outer long $h=3$\\$\Bcol\Bcol\Bcol$};
\node[loop,below=7mm of outer] (m1) {$M$};
\node[helix,fill=gray!8,below left=13mm and 22mm of m1] (critical) {critical long $h=3$\\$\Gray\Bcol\Bcol$};
\node[helix,fill=gray!8,below right=13mm and 22mm of m1] (side) {long $h=3$\\$\Bcol\Bcol\Bcol$};
\node[loop,below=7mm of critical] (m2) {$M$};
\node[helix,fill=orange!10,below left=12mm and 13mm of m2] (s1) {short\\$\Gray\Bcol$};
\node[helix,fill=orange!10,below right=12mm and 13mm of m2] (s2) {short\\$\Gray\Bcol$};
\draw[arr] (outer)--(m1);
\draw[arr] (m1)--node[left,font=\footnotesize]{$\Gray$} (critical);
\draw[arr] (m1)--node[right,font=\footnotesize]{$\Bcol$} (side);
\draw[arr] (critical)--node[right,font=\footnotesize]{closes $\Bcol$} (m2);
\draw[arr] (m2)--node[left,font=\footnotesize]{$\Gray$} (s1);
\draw[arr] (m2)--node[right,font=\footnotesize]{$\Gray$} (s2);
\node[right=17mm of m2,align=left] {$X(M)=\{\Wcol,\Gray,\Gray\}$\\\textcolor{GoodGreen}{proper}};
\end{tikzpicture}
\caption{An internal multiloop with two child subtrees containing isolated stacks.  The word $\Gray\Bcol\Bcol$ preserves the gray entry, reaches the $M$ residue, and closes non-gray, allowing both child helices to start gray.}
\label{fig:internal-example}
\end{figure}

Figure~\ref{fig:internal-example} shows the exact internal $r=2$ configuration realized by this deterministic coloring.

\subsection{What the examples do and do not show}

The examples demonstrate that the new $r=2$ root and multiloop cases are nonvacuous.  Their exact folding checks are regression tests, not evidence for the universal theorem.  Universality comes from the induction over arbitrary helix subtrees.

\clearpage
\section{Lean formalization and artifact verification}
\label{sec:formal}

The complete theorem was formalized in Lean~4 \citep{lean4} using Mathlib \citep{mathlib2020}.  The following listing is generated directly from the frozen source; it includes the exact public proposition, theorem, and proof term rather than Lean-like pseudocode.

\begin{lstlisting}[language={},caption={Exact final Lean proposition and theorem, extracted from the frozen source.},label={lst:public-theorem}]
def AtMostTwoShortHelixDesignabilityStatement : Prop :=
  ∀ {n : Nat} (T : SecondaryStructure n),
    InTargetClassKLeTwo T →
      ∃ w : Sequence n, UniqueDesigns w T
theorem atMostTwoShortHelixDesignability :
    AtMostTwoShortHelixDesignabilityStatement := by
  intro n T hK
  exact ⟨sequenceOfProperColoring
      (globalColoringCertificateLeTwo hK).coloring
      (globalColoringCertificateLeTwo hK).proper,
    atMostTwoShortHelices_uniqueDesigns hK⟩
\end{lstlisting}

The definition of \texttt{UniqueDesigns} quantifies over every \texttt{SecondaryStructure n} compatible with the same complete sequence and uses a strict pair-count inequality for every distinct competitor.  The target class uses the actual cardinality of the finite set of maximal helices of length 2.  A checked source map records the source path, Lean module, inclusive line range, full-source digest, and extracted-block digest of every declaration displayed here and in Appendix~\ref{app:lean}.  A generated audit module imports the real theorem module and successfully applies \texttt{\#check} to all 42 displayed declarations.

\subsection{Structure of the Lean development}

The Lean development follows the same sequence of ideas as the paper:
\begin{enumerate}
\item the four-letter alphabet, complementarity, arbitrary noncrossing partial matchings, interval tree, motifs, maximal helices, and the public unique-designability statement;
\item colors, exposures, levels, the uniform modulo-2 condition, exact helix transfers, and the decomposition of the target into disjoint helix subtrees;
\item the definitions of the feasible states $F$ and $Q$, the isolated-stack count decomposition, the strengthened non-gray-closing transfer, color assignments to the actual child helices, and the recursive subtree invariant;
\item deterministic sequence assignment, followed by the counting and uniqueness lemmas; and
\item composition of the new global coloring with the general coloring-to-design theorem.
\end{enumerate}
The old exact-one theorem is not a premise of the at-most-two result; it is derived afterward as a corollary.

\subsection{Concrete semantic regression theorems}

Two new downstream-only Lean modules pin the formal definitions to concrete publication examples without entering the dependency closure of Listing~\ref{lst:public-theorem}.  Kernel-checked theorems establish that $T_1$ and $T_2$ belong to \texttt{InTargetClassKLeTwo}; that their displayed literal sequences satisfy \texttt{UniqueDesigns}; that the $T_2$ coloring has helix words $\Bcol\Bcol\Bcol/\Gray\Bcol\Bcol/\Gray\Bcol/\Gray\Bcol/\Bcol\Bcol\Bcol$ and satisfies the named-residue condition with $\xi=0$; that the three-isolated-stack and length-1 examples are rejected by the target predicate; and that \texttt{AUAU} fails to uniquely design the nested target \texttt{(())}.  Ordinary kernel reduction over all $3^6=729$ colorings also proves that the three-isolated-stack target has no coloring that is both proper and strongly modulo-2 separated.

The literal design theorems do not call an external folding oracle.  They identify each literal word with the generic sequence produced from an explicit proper separated coloring, then apply the universal uniqueness theorem.  Their transitive axiom set is the same three foundational principles listed below.  Hash comparison confirms that all 60 pre-existing tracked Lean files remained byte-identical and that no reverse import from the theorem closure to either publication-example module was introduced.

\subsection{Axiom check and isolated rebuild}

Lean reports the following transitive axiom set for the final theorem:
\[
\{\texttt{propext},\ \texttt{Classical.choice},\ \texttt{Quot.sound}\}.
\]
The final theorem uses no project-defined axiom or unfinished proof placeholder.  In particular, its dependencies contain no \texttt{sorry}, \texttt{admit}, \texttt{sorryAx}, mathematical use of \texttt{unsafe}, \texttt{native\_decide}, or external proof program.

A frozen source archive was rebuilt in an isolated Linux container using the pinned Lean and Mathlib revisions.  The archived qualification rebuild completed 3,070 jobs, printed the exact final theorem, checked its foundational dependencies, searched the source for unfinished or externally supplied proof steps, and verified that compilation did not alter the source files.

The final publication branch was also qualified by a measured local clean build.  The authoritative 46-file theorem closure contains 20,738 physical lines, 16,908 nonblank noncomment code lines, and 1,616 source declaration commands: 390 \texttt{def}, 36 \texttt{abbrev}, 20 \texttt{structure}, 7 \texttt{inductive}, 1,102 \texttt{theorem}, and 61 \texttt{instance} commands.  After \texttt{lake clean}, the bare full-project command \texttt{lake build} passed with 3,070 Lake graph jobs in 894.10 seconds wall-clock on the recorded ten-logical-CPU host.  These are measured qualification values, not performance claims.  The two downstream publication modules add 613 physical lines and 56 declaration commands; their named build and 18-endpoint axiom audit also passed.  The \texttt{docs/FORMALIZATION\_METRICS.md} report states the counting method, environment, user and system times, and complete command log; its SHA-256 is \hashtext{b97a0465d97c9568211cf959e35569a4f82c766f511ca553fd1dd6ef1e1433f9}.

A separate anonymous package contained only the mathematical specification, the handwritten Lean files needed by the final theorem, and the version-pinned build files.  A fresh Claude Code session checked the file list and checksums, built the target module, traced the theorem dependencies, tested concrete examples and counterexamples, and compared the Lean definitions with the specification.  This was a blinded automated audit by a separate, fresh AI session, not an independent human review.  The audit report recorded the outcome ``faithful and complete.''  It records the following checks:
\begin{itemize}
\item exactly four nucleotides and strict Watson--Crick compatibility;
\item arbitrary noncrossing partial-match competitors on the same sequence;
\item the exact at-most-two target class, including the all-unpaired case;
\item the isolated-stack count decomposition and the inductive guarantees expressed by $F$ and $Q$;
\item the $h=3$ and $h=4$ strengthened transfers;
\item loop and root color assignments when two child subtrees contain isolated stacks;
\item that subtree colorings remain valid when the rest of the target is colored, and that every recursive call is made on a smaller subtree; and
\item the direct dependency of the final theorem on the new inductive coloring construction.
\end{itemize}

\begin{verificationbox}
\textbf{Validated release artifact identifiers.}
\begin{itemize}
\item canonical source archive: \hashtext{b077b6118a6cbfadb6c6fce00af65ed34125ffa431f6ca8b7e2e9de3c268d40f};
\item release-qualification results package: \hashtext{2da05881fb3808e7add46526f6f202f81e6f0b4f9b37acd622a8acd7b435615b};
\item blinded fidelity bundle: \hashtext{bfd2c8d90f5666ccb2dda6a3e9c2f81b84419ae84cc25beb3f880c80b5c66b8b};
\item Claude fidelity-audit package: \hashtext{0a77ea7e9f9dc0342a14330b93644dc943a0b2c0e89f302b092ca08de404092b}; and
\item manuscript source package: generated from this revision, with its digest recorded in the external release manifest.
\end{itemize}
Canonical repository: \RepositoryURL.  Archival DOI: \ArchiveDOI.  Version \ReleaseVersion{} uses Apache License 2.0 for Lean source and software tools, Creative Commons Attribution 4.0 International for manuscript text, original figures, and project documentation, and the STIX Font License / SIL Open Font License 1.1 terms for bundled fonts; \texttt{LICENSES.md} in the archive gives the authoritative path scopes.  The DOI was reserved before this PDF was built and identifies the immutable Zenodo Software record for version \ReleaseVersion.  The external SHA-256 manifest identifies the exact release bytes.
\end{verificationbox}

\subsection{What formal verification establishes}

Kernel acceptance establishes that the final theorem follows from the formal definitions and the listed foundational principles.  The isolated rebuild establishes reproducibility of the exact frozen source.  The automated specification comparison addresses, with lower evidentiary weight, the separate risk that Lean might have proved a subtly weakened or different theorem.  Because that comparison was performed by an AI system, its verdict is not a substitute for kernel checking, reproducibility, public inspection, or human peer review.  None of these checks establishes biological realism, novelty, importance, empirical utility, or correctness of claims outside the formal model.

\section{AI-assisted development, manuscript preparation, and audit trail}
\label{sec:ai}

Section~\ref{sec:ai-disclosure} gives the full AI-use disclosure and division of responsibility.  The present section records how that assistance entered the development and how the final result was audited.  A recurring difficulty was the distinction among the target structure, the sequence assigned to it, and alternative folds of that same sequence.  Recasting the argument in terms of base-pair use and alternative folds helped clarify the role of the coloring construction.

The development proceeded in several stages:
\begin{enumerate}
\item framing and refinement of the research problem, formal model, theorem statement, and quantifiers;
\item assistance with literature discovery and citation and provenance checks;
\item generation, revision, and stress-testing of candidate proofs by ChatGPT and Codex;
\item separate attempts by Claude Code to find counterexamples, missing cases, and hidden assumptions;
\item independently implemented programs for exact colorability, optimum-counting folding, and negative controls;
\item a staged Lean formalization with code generation, debugging, and source and axiom review after each major part;
\item development of examples and figures and drafting, restructuring, editing, proofreading, and typesetting of the manuscript; and
\item assembly of reproducibility and release artifacts, followed by an isolated rebuild and a blinded automated comparison of the mathematical specification with the Lean source.
\end{enumerate}

Several early versions contained genuine errors or incomplete case analyses, including an incorrect orientation rule for gray pairs and overly broad statements about long-helix transfers.  The evidence hierarchy for the final theorem is:
\begin{enumerate}
\item Lean kernel checking;
\item isolated clean rebuild;
\item source and axiom audits;
\item blinded automated fidelity audit; and
\item computational checks.
\end{enumerate}

\noindent\textbf{Review status and author responsibility.} At the time of this arXiv version, the manuscript and formal proof have not yet received a completed independent review by a human subject-matter expert. Generative-AI systems played the foundational and pervasive role described in Section~\ref{sec:ai-disclosure} throughout the research, formalization, verification, artifact preparation, figure development, and writing, under the human author's direction; the formal proof was verified in Lean. Lean kernel checking, isolated rebuilding of the frozen source, and automated audits address logical correctness, reproducibility, and fidelity to the stated formal model. They do not substitute for human expert review and do not establish novelty, significance, or biological relevance. The author takes full responsibility for the final scientific model, theorem statement, proof exposition, formalization, computations, examples and figures, citations, and all claims in the manuscript.

\section{Scope, limitations, and discussion}
\label{sec:scope}

\subsection{Model limitations}

The theorem is exact only in the stated base-pair model.  It excludes G--U wobble, pseudoknots, stacking energies, loop entropies, ionic effects, ensemble objectives, kinetics, and three-dimensional constraints.  Minimum arc length is $\theta=0$, so empty hairpins and adjacent pairs are legal.  If a target remains valid under a stricter minimum-loop convention, the same design remains unique because the competitor set only shrinks; however, some targets in $\Ktwo$ cease to be legal targets when $\theta>0$.

Avoidance of $m_{3\bullet}$ is biologically restrictive.  Any nonroot loop with at least two branch helices must contain no unpaired nucleotide in this model, and an exterior loop with an unpaired nucleotide may have at most two paired children.  The theorem therefore concerns a narrow but mathematically natural combinatorial class rather than typical full thermodynamic RNA architectures.

\subsection{What the bound two means}

The phrase ``at most two'' is a resource limit of the present universal construction, not a claimed structural boundary for designability or even for modulo-2 separability.  The theorem gives a uniform modulo-2 coloring for every target in $\Ktwo$: the stated motif-free class with no isolated base pair, at most two isolated stacks, and all remaining helices of length at least three.  Proposition~\ref{prop:three-short} shows only that a target with three isolated stacks can fail to admit a proper modulo-2 separated coloring.  It does not show that three isolated stacks always imply modulo-2 nonseparability, ordinary nonseparability, or undesignability.  Indeed, the displayed target is ordinarily separated and designable.  The natural next questions are therefore:
\begin{itemize}
\item which targets with three or more isolated stacks admit modulo-2 separated colorings;
\item whether a higher modulus, ordinary separation, or biseparability \citep{boury2025seeding} yields broader bounded-isolated-stack theorems;
\item whether other finite collections of feasible entry states yield broader guarantees; and
\item which structural placements, rather than counts alone, govern the transition from universal guarantees to obstructions.
\end{itemize}

The last question is the most informative formulation of the remaining problem: find a local placement criterion that distinguishes harmless collections of short helices from genuine coloring obstructions.  More broadly, constructive work on infinite unsaturated designable tree families and complementary work on undesignable motifs and asymptotic scarcity show that positive and negative structural phenomena coexist well beyond any one helix-count threshold \citep{jedwab2020,yao2026}.

\subsection{Novelty and related-work boundary}

The attribution boundary is as follows.  Proper separated colorings and their implication to designability originate with Hale\v{s} et al.; their Theorem~10 already uses opposite parity classes for gray and unpaired nodes, and their Result~R4 already covers the designability of saturated targets \citep{hales2017}.  Boury et al. introduced general modulo-$m$ separability, its $O(n\,2^m)$ dynamic program, and the universal all-long theorem \citep{boury2024,boury2025}.  Their WABI Figure~8 and journal Figure~12 also contain every local word family used in Lemma~\ref{lem:strong-long}, including the non-gray-ending $\Gray\Bcol\Bcol$ case.  They note that the modulo-2 method extends beyond the no-short-helix class ``in a way that remains to be fully characterized'' \citep{boury2024,boury2025}.

Subject to the literature-search limitation below, the new structural contribution here is therefore not the parity idea, the finite-state decision algorithm, or a new long-helix transfer word.  It is the direct resource-counting induction proving that targets with one or two isolated stacks always admit the required modulo-2 coloring: when two downstream short-helix demands coincide, they exhaust the global budget, forcing the incoming helix to be long and making Boury et al.'s published non-gray-ending option available.  For saturated members, only this certificate is new; their designability was already known.  Individual unsaturated targets may also belong to other previously known sufficient classes \citep{jedwab2020}.  A documented search through 19 August 2026 did not locate an exact prior theorem or counterexample for this bounded-isolated-stack class.  That search is necessarily incomplete and should not be treated as a proof of novelty.

\subsection{Algorithmic implications}

Boury et al.'s dynamic program already decides modulo-2 separability and constructs a witness in linear time, so this theorem neither supplies a previously unavailable decision capability nor improves the general asymptotic bound \citep{boury2024,boury2025}.  Its algorithmic content is a structural success guarantee: on every target in $\Ktwo$, the modulo-2 procedure cannot report nonseparability, and the direct proof gives an explicit top-down witness without searching the full state table.  Targets outside $\Ktwo$ can still be tested by the same fixed-$m$ algorithm, at $m=2$, $m=3$, or another fixed modulus; the present theorem provides no blanket guarantee for them.

The explicit coloring algorithm here runs in linear time when the tree representation and maximal helices are stored explicitly.  After fixing the root-residue convention, symmetric transfer choices, child order, and gray-pair orientations, the top-down rules supply a deterministic sequence.  This gives an exact constructive baseline for testing inverse-folding software and for initializing more realistic thermodynamic design methods.  Whether these sequences are good seeds under Turner energies is an empirical question; published work suggests that separated and biseparable sequences can be useful starting points, but the present theorem makes no thermodynamic claim \citep{boury2025,boury2025seeding}.

\section{Conclusion}
\label{sec:conclusion}

We have proved that every motif-free RNA target with no isolated base pair, at most two isolated stacks, and all other helices of length at least three is designable in the four-letter Watson--Crick base-pair model.  The proof uses the two sets of feasible entry states $F$ and $Q$ in an induction on helix subtrees.  If two child subtrees both require gray first pairs, they already contain the two isolated stacks allowed in the target.  The incoming helix is therefore long and can reach residue $\eta$ while ending non-gray, allowing both child helices to start gray.

The complete result, including the sequence construction and uniqueness theorem, has been formalized in Lean and rebuilt from frozen source.  A blinded automated AI audit separately compared the mathematical specification with the Lean source as a secondary cross-check.  The result gives a formally verified sufficient class in the $h_{\min}=2$ regime.  Its new step is the bounded-resource induction that deploys Hale\v{s} et al.'s parity idea and Boury et al.'s published local transfers in a global construction, not a new parity device or helix word.  It also illustrates how AI-assisted research, formalization, auditing, and manuscript preparation can be brought to a reproducible formal endpoint.

\appendix

\section{Local transition and color-assignment tables}
\label{app:tables}

\subsection{Two-pair transitions}

\begin{center}
\begin{tabular}{cc}
\toprule
entry $\to$ exit & valid words \\
\midrule
$\xi\to\xi$ & $\Bcol\Bcol,\Wcol\Wcol$ \\
$\xi\to\eta$ & $\Bcol\Gray,\Wcol\Gray$ \\
$\eta\to\xi$ & $\Gray\Bcol,\Gray\Wcol$ \\
$\eta\to\eta$ & $\Bcol\Bcol,\Wcol\Wcol,\Gray\Gray$ \\
\bottomrule
\end{tabular}
\end{center}

\subsection[Multiloop assignments]{Multiloop color assignments}

\begin{center}
\begin{tabular}{cccc}
\toprule
child subtrees containing isolated stacks $r$ & closing color & $d=2$ & $d=3$ \\
\midrule
0 & $\Bcol$ & $\Bcol,\Gray$ & $\Bcol,\Gray,\Gray$ \\
0 & $\Wcol$ & $\Wcol,\Gray$ & $\Wcol,\Gray,\Gray$ \\
0 & $\Gray$ & $\Bcol,\Wcol$ & $\Bcol,\Wcol,\Gray$ \\
\midrule
1 & $\Bcol$ & $\Gray,\Bcol$ & $\Gray,\Bcol,\Gray$ \\
1 & $\Wcol$ & $\Gray,\Wcol$ & $\Gray,\Wcol,\Gray$ \\
1 & $\Gray$ & $\Gray,\Bcol$ & $\Gray,\Bcol,\Wcol$ \\
\midrule
2 & $\Bcol$ & $\Gray,\Gray$ & $\Gray,\Gray,\Bcol$ \\
2 & $\Wcol$ & $\Gray,\Gray$ & $\Gray,\Gray,\Wcol$ \\
\bottomrule
\end{tabular}
\end{center}
In the $r=1$ and $r=2$ rows, gray entries are assigned to the children whose subtrees contain an isolated stack; displayed order is schematic.

\subsection{Root color assignments at degree three and four}

\begin{center}
\begin{tabular}{cccc}
\toprule
root degree & $r=0$ & $r=1$ & $r=2$ \\
\midrule
3 & $\Bcol,\Wcol,\Gray$ & $\Gray,\Bcol,\Wcol$ & $\Gray,\Gray,\Bcol$ \\
4 & $\Bcol,\Wcol,\Gray,\Gray$ & $\Gray,\Bcol,\Wcol,\Gray$ & $\Gray,\Gray,\Bcol,\Wcol$ \\
\bottomrule
\end{tabular}
\end{center}

\section{Selected Lean definitions}
\label{app:lean}

The following declarations are exact Unicode excerpts generated from the frozen source, not simplified retypings.  Blank lines are the only inserted material between source blocks.  In addition to the top-level model and target predicate, the selection now prints the interval-tree node types, parent and child operations, paired degree, unpaired-child predicate, and exact stack-offset relation on which the $m_5$, $m_{3\bullet}$, and maximal-helix encodings depend.  The accompanying source map and generated \texttt{\#check} audit are distributed with the manuscript source.

\begin{lstlisting}[language={},caption={Exact core-model declarations extracted from the frozen Lean source.}]
inductive Nucleotide where
  | A
  | C
  | G
  | U
  deriving DecidableEq, Repr
def comp : Nucleotide → Nucleotide
  | A => U
  | U => A
  | C => G
  | G => C
def Compatible (x y : Nucleotide) : Prop := y = x.comp
abbrev Sequence (n : Nat) := Fin n → Nucleotide
structure Arc (n : Nat) where
  left : Fin n
  right : Fin n
  ordered : left < right
  deriving DecidableEq
def Crosses (a b : Arc n) : Prop :=
  (a.left < b.left ∧ b.left < a.right ∧ a.right < b.right) ∨
  (b.left < a.left ∧ a.left < b.right ∧ b.right < a.right)
def IsPartialMatching (P : Finset (Arc n)) : Prop :=
  ∀ {a : Arc n}, a ∈ P →
    ∀ {b : Arc n}, b ∈ P →
      a.ShareEndpoint b → a = b
def IsNoncrossing (P : Finset (Arc n)) : Prop :=
  ∀ {a : Arc n}, a ∈ P →
    ∀ {b : Arc n}, b ∈ P →
      ¬ a.Crosses b
structure SecondaryStructure (n : Nat) where
  arcs : Finset (Arc n)
  isPartialMatching : IsPartialMatching arcs
  isNoncrossing : IsNoncrossing arcs
def StructureCompatible (w : Sequence n) (S : SecondaryStructure n) : Prop :=
  ∀ a ∈ S.arcs, Compatible (w a.left) (w a.right)
def pairCount (S : SecondaryStructure n) : Nat := S.arcs.card
def UniqueDesigns (w : Sequence n) (T : SecondaryStructure n) : Prop :=
  StructureCompatible w T ∧
    ∀ S : SecondaryStructure n,
      StructureCompatible w S →
      S ≠ T →
      pairCount S < pairCount T
\end{lstlisting}

\begin{lstlisting}[language={},caption={Exact helix and target-class declarations extracted from the frozen Lean source.}]
abbrev PairedNode := {a : Arc n // a ∈ T.arcs}
abbrev UnpairedPosition := {i : Fin n // ¬ T.positionPaired i}
abbrev NonRootNode := PairedNode T ⊕ UnpairedPosition T
abbrev PairedOrRootNode := Option (PairedNode T)
def StrictlyContains (p : PairedNode T) : NonRootNode T → Prop
  | Sum.inl q => p.val.left < q.val.left ∧ q.val.right < p.val.right
  | Sum.inr u => p.val.left < u.val ∧ u.val < p.val.right
def enclosingPairs (c : NonRootNode T) : Finset (PairedNode T) :=
  Finset.univ.filter (fun p => PairedNode.StrictlyContains T p c)
def parent (c : NonRootNode T) : PairedOrRootNode T :=
  if h : (enclosingPairs T c).Nonempty then
    some ((enclosingPairs T c).max' h)
  else
    none
def pairedChildren (p : PairedOrRootNode T) : Finset (PairedNode T) :=
  Finset.univ.filter (fun c => parent T (Sum.inl c) = p)
def unpairedChildren (p : PairedOrRootNode T) : Finset (UnpairedPosition T) :=
  Finset.univ.filter (fun c => parent T (Sum.inr c) = p)
def pairedChildCount (p : PairedOrRootNode T) : Nat :=
  (pairedChildren T p).card
def pairedDegree : PairedOrRootNode T → Nat
  | none => pairedChildCount T none
  | some p => 1 + pairedChildCount T (some p)
def Arc.StackOffset (outer inner : Arc n) (k : Nat) : Prop :=
  inner.left.val = outer.left.val + k ∧
    inner.right.val + k = outer.right.val
def Arc.Stacked (outer inner : Arc n) : Prop :=
  outer.StackOffset inner 1
abbrev HelixCandidate (n : Nat) := Arc n × Fin (n + 1)
def outer (H : HelixCandidate n) : Arc n := H.1
def length (H : HelixCandidate n) : Nat := H.2.val
def HasStackOffset (T : SecondaryStructure n) (outer : Arc n) (k : Nat) : Prop :=
  ∃ a ∈ T.arcs, outer.StackOffset a k
def IsMaximalHelixRun (T : SecondaryStructure n) (outer : Arc n) (h : Nat) : Prop :=
  0 < h ∧
    outer ∈ T.arcs ∧
    (¬ ∃ a ∈ T.arcs, a.Stacked outer) ∧
    (∀ k : Fin h, HasStackOffset T outer k.val) ∧
    ¬ HasStackOffset T outer h
def IsMaximalHelix (T : SecondaryStructure n) (H : HelixCandidate n) : Prop :=
  IsMaximalHelixRun T H.outer H.length
abbrev MaximalHelix (T : SecondaryStructure n) :=
  {H : HelixCandidate n // IsMaximalHelix T H}
def outer {T : SecondaryStructure n} (H : MaximalHelix T) : Arc n := H.val.outer
def length {T : SecondaryStructure n} (H : MaximalHelix T) : Nat := H.val.length
def HasUnpairedChild (p : PairedOrRootNode T) : Prop :=
  (unpairedChildren T p).Nonempty
def HasM5 : Prop :=
  ∃ p : PairedOrRootNode T, 4 < pairedDegree T p
def HasM3Dot : Prop :=
  ∃ p : PairedOrRootNode T,
    HasUnpairedChild T p ∧ 2 < pairedDegree T p
def lengthTwoHelices (T : SecondaryStructure n) : Finset (MaximalHelix T) :=
  maximalHelicesOfLength T 2
def shortHelixCount (T : SecondaryStructure n) : Nat :=
  (lengthTwoHelices T).card
def InTargetClassKLeTwo (T : SecondaryStructure n) : Prop :=
  shortHelixCount T ≤ 2 ∧
    (∀ H : MaximalHelix T, H.length ≠ 1) ∧
    (∀ H : MaximalHelix T, H.length ≠ 2 → 3 ≤ H.length) ∧
    ¬ HasM5 T ∧
    ¬ HasM3Dot T
\end{lstlisting}

\section{Formal verification record and artifacts}
\label{app:artifacts}

\begin{longtable}{p{0.29\linewidth}p{0.61\linewidth}}
\toprule
Artifact & Record \\
\midrule
\endhead
Canonical proof specification & \texttt{CANONICAL\_AT\_MOST\_TWO\_PROOF.md}; SHA-256 \hashtext{fb9bee1126791b9ca53b3903c061df27e3856dd20b9c7fc248b7acd4e5326841}. \\
Canonical source archive & SHA-256 \hashtext{b077b6118a6cbfadb6c6fce00af65ed34125ffa431f6ca8b7e2e9de3c268d40f}. \\
Lean toolchain & Lean \texttt{4.34.0-rc1}; Mathlib revision \texttt{de5ce8a9a66a4aa68a9bdbb35b63a06d34d9ca11}. \\
Isolated rebuild & Full clean build passed; final theorem printed with axiom set \texttt{[propext, Classical.choice, Quot.sound]}. \\
Release-qualification results package & SHA-256 \hashtext{2da05881fb3808e7add46526f6f202f81e6f0b4f9b37acd622a8acd7b435615b}. \\
Blinded fidelity bundle & SHA-256 \hashtext{bfd2c8d90f5666ccb2dda6a3e9c2f81b84419ae84cc25beb3f880c80b5c66b8b}. \\
Claude fidelity-audit package & Automated-AI outcome ``faithful and complete''; SHA-256 \hashtext{0a77ea7e9f9dc0342a14330b93644dc943a0b2c0e89f302b092ca08de404092b}. \\
Publication examples & Downstream modules \texttt{PublicationExamples.lean} and \texttt{PublicationExamplesAxiomAudit.lean}; no reverse import and no change to any pre-existing Lean file. \\
Formalization-metrics report & Clean full-project build: 3,070 Lake graph jobs and 894.10 seconds wall-clock; closure: 46 files, 20,738 physical lines, 1,616 declaration commands; \texttt{docs/FORMALIZATION\_METRICS.md} SHA-256 \hashtext{b97a0465d97c9568211cf959e35569a4f82c766f511ca553fd1dd6ef1e1433f9}. \\
Derived listings & \texttt{SOURCE\_MAP.tsv}, extraction script, and 42-declaration \texttt{\#check} audit included in the manuscript source package. \\
Manuscript source package & Generated from this revision; SHA-256 recorded in the external release manifest to avoid a circular self-digest. \\
Availability & Version \ReleaseVersion; canonical repository \RepositoryURL; archival DOI \ArchiveDOI; mixed license per \texttt{LICENSES.md}. \\
\bottomrule
\end{longtable}

\section*{Author declarations}
\textbf{Author capacity.} This work was undertaken through Szilard Scientific, LLC, in the author's private capacity.  The views expressed are the author's own, not those of any employer or client.

\section*{Acknowledgments}
The author thanks the developers and communities of Lean and Mathlib.  OpenAI's ChatGPT and Codex and Anthropic's Claude Code played the foundational role detailed in Section~\ref{sec:ai-disclosure} under the human author's direction.  They are not authors and bear no responsibility for the scientific claims.

\section*{Data and code availability}
The version \ReleaseVersion{} archival release contains the manuscript source and PDF, source-derived Lean listings and their audit, deterministic example checker and output, citation and provenance audits, downstream Lean publication examples, and frozen qualification packages.  Appendix~\ref{app:artifacts} records the canonical repository, archival DOI, and mixed-license scope; the external SHA-256 manifest identifies the exact release bytes.

\end{document}